\documentclass[sigconf]{acmart}

\usepackage{subfiles}
\usepackage{subfigure}
\usepackage{amsmath}
\usepackage{algorithmic}
\usepackage[ruled,vlined,linesnumbered]{algorithm2e}
\usepackage{blindtext}
\usepackage[font=small]{caption}
\usepackage{comment}
\usepackage{enumitem}
\usepackage{float}
\usepackage{framed}
\usepackage{graphicx}
\usepackage{xargs}
\usepackage{marginnote}
\usepackage{makecell}
\usepackage{multirow}
\usepackage{colortbl}
\usepackage{textcomp}
\usepackage{soul}

\usepackage[flushleft]{threeparttable}
\newcommand{\nop}[1]{}

\newtheorem{lemma}{Lemma}

\newtheorem{theorem}{Theorem}
\newtheorem{fact}{Fact}
\theoremstyle{definition}

\newtheorem{definition}{Definition}

\newtheorem{remark}{Remark}
\allowdisplaybreaks

\makeatletter
\newsavebox\myboxA
\newsavebox\myboxB
\newlength\mylenA

\newcommand*\xbar[2][0.75]{%
    \sbox{\myboxA}{$\m@th#2$}%
    \setbox\myboxB\null% Phantom box
    \ht\myboxB=\ht\myboxA%
    \dp\myboxB=\dp\myboxA%
    \wd\myboxB=#1\wd\myboxA% Scale phantom
    \sbox\myboxB{$\m@th\overline{\copy\myboxB}$}%  Overlined phantom
    \setlength\mylenA{\the\wd\myboxA}%   calc width diff
    \addtolength\mylenA{-\the\wd\myboxB}%
    \ifdim\wd\myboxB<\wd\myboxA%
       \rlap{\hskip 0.5\mylenA\usebox\myboxB}{\usebox\myboxA}%
    \else
        \hskip -0.5\mylenA\rlap{\usebox\myboxA}{\hskip 0.5\mylenA\usebox\myboxB}%
    \fi}
\makeatother

\newcommand{\vect}[1]{\ensuremath{\mathbf{#1}}}

\newcommand{\grad}{\nabla}

\newcommand{\g}{\vect{g}}

\newcommand{\x}{\vect{x}}

\AtBeginDocument{%
  }

\setcopyright{acmlicensed}
\copyrightyear{2024}
\acmYear{2024}
\setcopyright{acmlicensed}\acmConference[CCS '24]{Proceedings of the 2024
ACM SIGSAC Conference on Computer and Communications Security}{October
14--18, 2024}{Salt Lake City, UT, USA}
\acmBooktitle{Proceedings of the 2024 ACM SIGSAC Conference on Computer
and Communications Security (CCS '24), October 14--18, 2024, Salt Lake City,
UT, USA}
\acmDOI{10.1145/3658644.3670351}
\acmISBN{979-8-4007-0636-3/24/10}
\begin{document}

\title{Cross-silo Federated Learning with Record-level~Personalized~Differential~Privacy}

\author{Junxu Liu}
\affiliation{
    \institution{Renmin University of China}
    \city{Beijing}
    \country{China}}
\email{geminiljx@gmail.com}

% 2nd. author
\author{Jian Lou}
\affiliation{
    \institution{Zhejiang University}
    \city{Hangzhou}
    \country{China}}
\email{jian.lou@zju.edu.cn}

\author{Li Xiong}
\affiliation{
    \institution{Emory University}
    \city{Atlanta}
    \country{USA}}
\email{lxiong@emory.edu}

\author{Jinfei Liu}
\affiliation{
    \institution{Zhejiang University}
    \city{Hangzhou}
    \country{China}}
\email{jinfeiliu@zju.edu.cn}

\author{Xiaofeng Meng}
\authornote{Corresponding author: Xiaofeng Meng.}
\affiliation{
    \institution{Renmin University of China}
    \city{Beijing}
    \country{China}}
\email{xfmeng@ruc.edu.cn}

\renewcommand{\shortauthors}{Junxu Liu, Jian Lou, Li Xiong, Jinfei Liu, and Xiaofeng Meng.}

\begin{abstract}
Federated learning (FL) enhanced by differential privacy has emerged as a popular approach to better safeguard the privacy of client-side data by protecting clients' contributions during the training process. Existing solutions typically assume a uniform privacy budget for all records and provide one-size-fits-all solutions that may not be adequate to meet each record's privacy requirement. 
In this paper, we explore the uncharted territory of cross-silo FL with record-level personalized differential privacy.
We devise a novel framework named \textit{rPDP-FL}, employing a two-stage hybrid sampling scheme with both uniform client-level sampling and non-uniform record-level sampling to accommodate varying privacy requirements. 

A critical and non-trivial problem is how to determine the ideal per-record sampling probability $q$ given the personalized privacy budget $\varepsilon$. 
We introduce a versatile solution named \textit{Simulation-CurveFitting}, allowing us to uncover a significant insight into the nonlinear correlation between $q$ and $\varepsilon$ and derive an elegant mathematical model to tackle the problem. Our evaluation demonstrates that our solution can provide significant performance gains over the baselines that do not consider personalized privacy preservation.
\end{abstract}

\begin{CCSXML}
<ccs2012>
   <concept>
       <concept_id>10002978.10003029.10011150</concept_id>
       <concept_desc>Security and privacy~Privacy protections</concept_desc>
       <concept_significance>500</concept_significance>
       </concept>
 </ccs2012>
\end{CCSXML}

\ccsdesc[500]{Security and privacy~Privacy protections}

\keywords{Federated Learning, Differential Privacy, Personalized Privacy Protection}


\maketitle

\section{Introduction}

Federated Learning (FL) \cite{mcmahan2017communication,kairouz2021advances} is a recent machine learning (ML) framework that was motivated by data privacy. In comparison to centralized ML, it eliminates the need for centralized data sharing and has the potential to harness decentralized data for powerful predictive models while alleviating individual privacy concerns. The distinctive feature is its decentralized architecture, where multiple institutions (e.g., hospitals or banks) or devices (e.g., smartphones, IoT devices \cite{DBLP:conf/mobisys/ChenWZQHWFP23,DBLP:journals/pvldb/GaoCLXPLDZ23,ma2021communication}) collaborate in training a joint model under the coordination of a central \textit{server} while keeping the data local. %Depending on the type of clients and the scale of model training, FL can be classified into two settings: cross-silo and cross-device. 
This paper primarily focuses on the former case, also known as cross-silo FL \cite{liu2022on}, where each \textit{client} (institution) holds a local dataset comprising personal data \textit{records}. For simplicity, we assume each \textit{record} is associated with a single \textit{user} (e.g. patient or customer) who does not contribute the same record or multiple records to multiple clients simultaneously.

Although data are not directly shared in FL, potential adversaries (e.g., the honest-but-curious server or untrusted clients) might engage in indirect privacy violations via reconstruction or inference attacks \cite{DBLP:conf/aistats/LowyGR23,hitaj2017deep,rigaki2023survey,tolpegin2020data,wang2019beyond,zhu2019deep,zhang2023survey,zhang2023moda}. Differential privacy (DP), known as the de facto standard for private data analysis, has been introduced to FL algorithm design \cite{liu2023echo,xiang2023practical,DBLP:journals/pvldb/LiWL23,DBLP:conf/ccs/MaddockC0MJ22, fu2024dpsur,malekmohammadi2024noise,fu2024differentially,liu2024federated}. This integration ensures rigorous privacy protection for participants (clients or records) by introducing controlled perturbation into the computation of the intermediate model parameters transferred between clients and the server \cite{mcmahan2017learning}. While standard DP provides the means to quantify the extent of privacy protection through a positive real-valued parameter $\varepsilon$ (aka \textit{privacy budget}), it imposes identical privacy safeguards on every participant involved. This uniformity cannot reflect the reality of diverse privacy expectations among people and can lead to significant utility costs \cite{jorgensen2015conservative, ebadi2015differential,boenisch2023have}. It is desirable to allow each participant to set their expected privacy budgets reflecting their personal privacy preferences. 

With this objective in mind, personalized differential privacy (PDP) \cite{jorgensen2015conservative, ebadi2015differential} was introduced and has been investigated in various scenarios including statistical analysis \cite{jorgensen2015conservative, ebadi2015differential, chen2016private}, centralized machine learning (ML) \cite{feldman2021individual, boenisch2023have}, and federated learning \cite{liu2021projected}. %Nevertheless, achieving PDP specifically in the context of FL has not received sufficient attention within the research community. 
For FL, Liu et al. \cite{liu2021projected} proposed the concept of heterogeneous DP in FL, where records within a single client (institution) share the same privacy budget, but different clients may have varying privacy budgets. We refer to this specific setting as \textit{client-level} PDP-FL for clarity. In contrast, this paper introduces a broader setting where even records within the same client may have distinctive privacy preferences, referred to as \textit{record-level} PDP-FL (rPDP-FL in short). %This will be the main focus of this paper. 
Figure \ref{fig:pdp_fl_framework} provides an illustrative example of the latter case in a healthcare context.
To the best of our knowledge, \textit{record-level} PDP-FL has not yet been investigated. 

\begin{figure}[t]\centering
\includegraphics[width=0.9\linewidth]{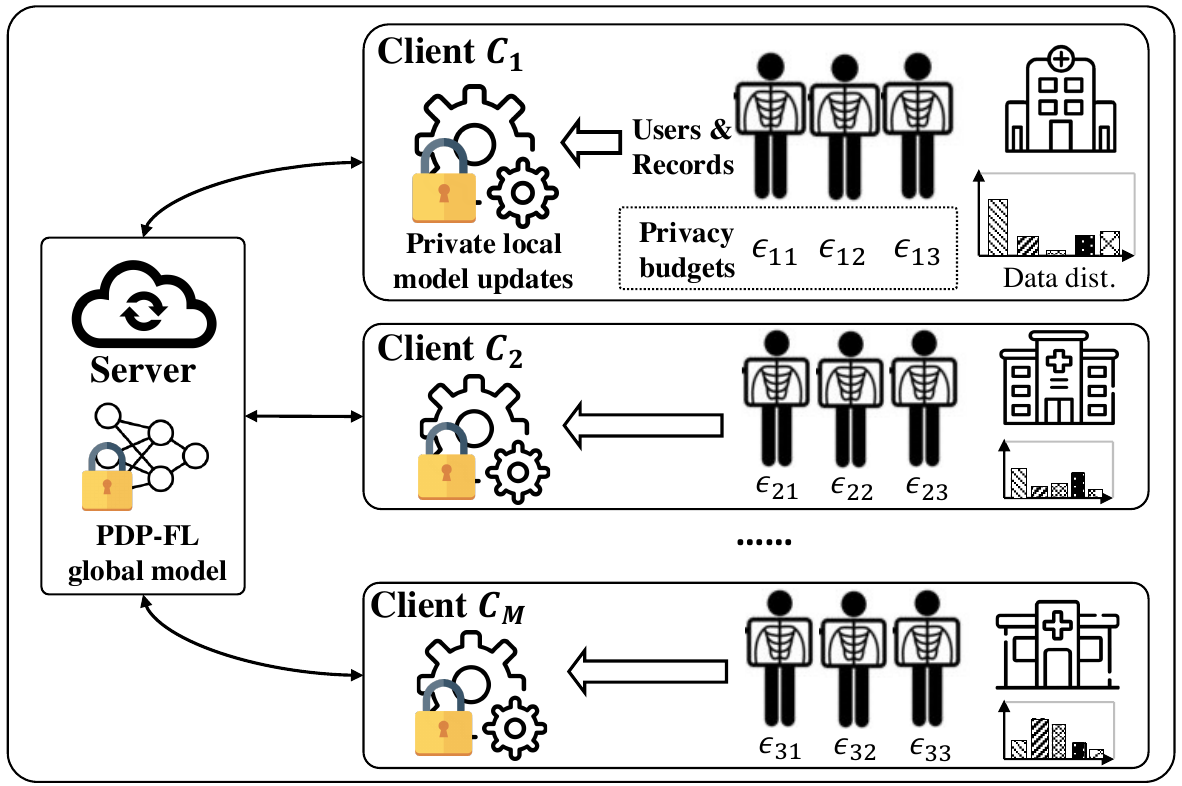}
\caption{An illustration of the cross-silo federated learning with record-level personalized differential privacy. In this framework, each user is given the autonomy to independently opt for a personalized privacy preference (specified by a personalized differential privacy (PDP) budget $\varepsilon$) for their respective records. The goal is to train a private global model that satisfies record-level PDP.}
\Description{An illustration of the cross-silo federated learning with record-level personalized differential privacy. In this framework, each user is given the autonomy to independently opt for a personalized privacy preference (estimated by a corresponding privacy budget $\varepsilon$) for their respective records. The goal is to train a private global model that satisfies record-level PDP.}
\label{fig:pdp_fl_framework}
\end{figure}

From a technical standpoint, the essence of implementing \textit{record-level} PDP lies in ensuring each record's accumulative privacy cost aligns with its predetermined privacy budget during the entire training process. This emphasizes the need for an effective privacy budget allocation strategy. When it comes to achieving PDP in a centralized ML setting (centralized PDP), some studies \cite{rogers2016privacy,feldman2021individual} use even privacy budgets across all records during every iteration of training, and records with smaller privacy budgets will be filtered out of the training process once their privacy budgets run out. This approach may trigger \textit{catastrophic forgetting} \cite{kirkpatrick2017overcoming,goodfellow2013empirical}, a situation where the learned model could potentially ``forget'' the knowledge from records that terminate early, eventually leading to degraded model performance. 
A more promising strategy involves achieving simultaneous depletion of privacy budgets for all records by developing DP mechanisms coupled with non-uniform sampling \cite{jorgensen2015conservative, boenisch2023have}. The fundamental theory underlying it is the ``privacy amplification by random sampling'' theorem \cite{li2012sampling,abadi2016deep,balle2018privacy,wang2019subsampled,zhu2019poission} which implies individuals with lower inclusion (sampling) probabilities $q$ will incur less privacy cost (leading to smaller privacy budget). 
Boenisch et al. \cite{boenisch2023have} introduced a binary search-based approach to determine an approximate optimal $q\in [0,1]$ for each record with a specific privacy budget $\varepsilon$. The computation of the accumulative privacy cost for each record is based on Mironov et al. \cite{mironov2019renyi} in which \textit{R\'enyi Differential Privacy} (RDP) \cite{mironov2017renyi,mironov2019renyi} is employed for tight privacy accounting. However, their approach is impractical when all records' privacy budgets are distributed continuously and can cover a spectrum of values, due to the scalability and computational costs associated with the per-record search process.

Overall, there exist significant research gaps in achieving record-level PDP-FL. First, the existing theoretical findings on RDP-based privacy accounting \cite{mironov2019renyi,zhu2019poission} are no longer adequate for the needs of analyzing each record's accumulative privacy cost in FL applications. Intuitively, the two-stage sampling process, i.e.,  client-level and record-level sampling, will further increase uncertainty for potential adversaries to infer whether a ``target'' record is a member of a client and hence further amplify the privacy protection. 
Second, the existing binary search-based approach \cite{boenisch2023have} for finding the optimal sampling probability $q$ in centralized PDP is not efficient. A desirable way is to directly compute a sampling probability given a privacy budget for each record. However, it's non-trivial to derive an explicit closed-form solution due to the highly nonlinear and less interpretable RDP-based privacy accounting. 

\noindent\textbf{Contributions.} Our key contributions are outlined as follows.
\begin{enumerate}[leftmargin=*]
\item We formalize a real-world problem in federated learning concerning record-level personalized differential privacy. To solve this problem, we propose a novel framework called \textit{rPDP-FL}. The essence of this framework is a \textit{two-stage hybrid sampling} scheme which comprises a uniform client-level sampling process and a non-uniform record-level sampling process. Specifically, the per-client sampling probability is assumed as a hyperparameter and publicly known by both the server and all clients, while the per-record sampling probability is proportional to each record's privacy budget and determined by the client to which it belongs.
\item We formally analyze the enhanced privacy amplification effect of the two-stage hybrid sampling scheme. This RDP-based theoretical investigation fills a gap in existing research and facilitates a more favorable trade-off between privacy and utility.
\item We devise an efficient and general strategy named \textit{Simulation-CurveFitting} (SCF) to identify the sampling probabilities for all records given their personalized privacy budgets. Our simulations with varying sampling probabilities enable the identification of an elegant mathematical function discerning the relationship between per-record sampling probabilities and their accumulative privacy costs. An important insight arises: the tight upper bound on the accumulative privacy cost of rPDP-FL can be modeled by a simple exponential function w.r.t. its record-level sampling probability. 
\item We simulate three potential personalized privacy scenarios and conduct a comprehensive evaluation on two real-world datasets. We first show that our SCF strategy outperforms the existing PDP methods for centralized ML \cite{boenisch2023have,feldman2021individual} in model utility and computational efficiency. Additionally, we demonstrate that rPDP-FL significantly enhances the utility of the global model compared to baseline methods that do not incorporate personalized privacy preservation.
\end{enumerate}

\section{Related Work}\label{sec:related_work}

\noindent\textbf{Personalized Differential Privacy.} 
The concept of personalized DP (PDP) was initially introduced by Ebadi et al. \cite{ebadi2015differential} and Jorgensen et al. \cite{jorgensen2015conservative}, focusing on basic private statistical analysis tasks with the standard $\varepsilon$-DP framework. Notably, the \textit{Sample} mechanism proposed in \cite{jorgensen2015conservative} demonstrated the feasibility of implementing PDP by combining DP mechanisms (e.g., Laplace or Gaussian mechanism) with non-uniform record-level sampling. 

Recent work studied PDP in centralized ML \cite{boenisch2023have} built on top of the non-uniform sampling strategy and proposed a binary search-based approach to find a suitable sampling probability $q$ as a decimal value within the range of $[0,1]$ for each record given a target privacy budget $\varepsilon$. It is, however, computationally demanding for the more realistic settings where records' privacy budgets are distributed continuously (e.g., follow a Gaussian or Pareto distribution) and can cover a range of values. 
Another line of work \cite{rogers2016privacy,feldman2021individual} considered all records' privacy budgets to be uniform during each iteration of the training process. Two individual privacy accounting techniques named \textit{privacy odometer} and \textit{privacy filter} are designed to monitor and restrict accumulative privacy costs for individual records throughout the training process so that a record will be excluded from the subsequent training iterations once its privacy budget is exhausted. This poses a potential risk \textcolor{black}{of} \textit{catastrophic forgetting} \cite{kirkpatrick2017overcoming,goodfellow2013empirical} and may lead to downgraded model performance. 

\noindent\textbf{Federated Learning with DP and Personalized DP.} We discuss existing work on FL with DP in two aspects: (1) the granularity of the DP guarantee, i.e., what information is protected (each client or each record), and (2) the level of personalization for DP, i.e., who has the right to specify the privacy budget (the central server, each client, or each record).

\begin{itemize}[leftmargin=*]
\item \textit{Client- vs. record-level privacy protection}.
There is rich literature exploring the DP-FL framework concerning potential adversaries. Specifically, these adversaries may be either solely recipients of the global model parameters (i.e., the other \textit{untrusted} clients or third parties) or recipients of local model updates (i.e. the \textit{honest-but-curious} central server). Within this framework, two categories of DP guarantees are recognized: client- and record-level DP. The former is achieved by adding random Gaussian noise to the aggregated local model updates to hide a single client’s contribution \cite{geyer2017differentially}, while the latter requires clients to perturb their computed gradients locally to obscure a single record's contribution \cite{mcmahan2017learning,wei2020federated,liu2022on}. Our primary focus lies on achieving record-level protection against both attack scenarios.
\item \textit{Client- vs. record-level privacy personalization}.
As mentioned earlier, the majority of approaches offer uniform privacy guarantees for all records involved, based on the one-sided considerations of the central server. Only a few studies recognize the necessity of privacy personalization within FL applications. 
Liu et al. \cite{liu2021projected} introduced the concept of heterogeneous DP and developed a projection-based framework to accommodate diverse privacy budgets among different clients. Although the work \cite{liu2022on} also proposed a similar notion known as \textit{silo-specific sample-level} DP, it primarily focused on addressing data heterogeneity challenges and did not address varying privacy needs. Liu et al. \cite{liu2023echo}, on the other hand, centered on cross-device FL and achieved personalized local differential privacy (PLDP) for clients' local model gradients. However, it requires a large number of clients for reasonable utility. Our research represents the first attempt to explore record-level privacy personalization in cross-silo FL.
\end{itemize}

\noindent\textbf{Tight Privacy Analysis for DP-FL.}
Conducting a tight analysis of the accumulative privacy cost is crucial for designing DP algorithms effectively. The predominant focus of research on this issue centers around centralized ML \cite{abadi2016deep, mironov2019renyi, zhu2019poission, wang2019subsampled}, with limited attention directed towards the FL scenarios \cite{noble2022differentially,girgis2021shuffled} where the employment of both data and client sampling may lead to an enhanced privacy amplification effect. 
Girgis et al. \cite{girgis2021shuffled} investigated a related issue but focused on offering local differential privacy (LDP) guarantees for clients' gradients. In their framework, only one step of the local Stochastic Gradient Descent (SGD) update is executed per client per round, whereas our algorithm allows for multiple local updates.
Noble et al. \cite{noble2022differentially} adopted RDP to track the privacy cost over the local SGD iterations, while using $(\varepsilon,\delta)$-DP to evaluate privacy costs over global communication rounds. This conventional privacy notion is often considered suboptimal in practical applications. 
Our work extends existing findings by leveraging RDP tools to estimate the gain of privacy caused by client sampling, see Section \ref{sec:privacy}.

\section{Preliminaries}

Differential Privacy (DP) is a robust and mathematically rigorous definition of privacy. It allows for the quantification of the information leaked by an algorithm about its input data. Note that when two datasets $D$ and $D'$ differ by only one record\footnote{In this work, we consider the presence/absence model of privacy, where protection is w.r.t. the presence/absence of a record in the analyzed dataset, e.g., $D'\triangleq D\setminus \{d\}$, instead of the replacement of a record with another.}, denoted as $D\sim D'$, we refer to them as adjacent datasets. 
\begin{definition} [($\varepsilon,\delta$)-Differential Privacy \cite{dwork2006calibrating, dwork2014algorithmic}] \label{def:dp}
A randomized algorithm $\mathcal{A}:\mathbb{D}\rightarrow\mathbb{O}$ satisfies ($\varepsilon,\delta$)-DP if for any pair of adjacent datasets $D, D'\in \mathbb{D}$ and any subsets of outputs $o\subseteq \mathbb{O}$, it holds that
$$\Pr[\mathcal{A}(D) \in o] \leq e^\varepsilon \Pr[\mathcal{A}(D') \in o] + \delta.$$
\end{definition}
The privacy guarantee is controlled by the ``privacy budget'' $\varepsilon>0$ and the parameter $\delta\geq 0$ which captures the probability that the pure $\varepsilon$-DP (i.e., ($\varepsilon,0$)-DP) is broken. While the standard ($\varepsilon,\delta$)-DP is widely used in a broad range of literature, it may not be suitable for some settings.  The following are two notable limitations associated with ($\varepsilon,\delta$)-DP recognized in literature: 
\begin{enumerate}[leftmargin=*]
    \item ($\varepsilon,\delta$)-DP provides \textit{uniform} privacy guarantees for the entire dataset regardless of the individuals' preferences; 
    \item ($\varepsilon,\delta$)-DP offers a relatively \textit{loose} composition bound and thus it is not suitable to track and analyze the overall privacy cost of complex iterative algorithms which will lead to poor privacy and utility trade-off.
\end{enumerate}

In this study, our aim is to design a finely tailored algorithm with personalized privacy that effectively tackles the aforementioned challenges in the context of FL applications. We first review the notions of personalized differential privacy (PDP) and R\'enyi differential privacy (RDP), both of which form building blocks for our privacy analysis and algorithm design. More specifically, PDP tailors the level of privacy protection based on the specific privacy preferences of each record.  RDP offers a versatile framework for tight privacy accounting and better privacy-utility trade-offs.

\subsection{Personalized Differential Privacy}

Personalized DP is a variation of DP that bounds the \textit{individual} privacy cost for each record in the dataset. For example, the privacy guarantee for a specific record $d_j$ is defined over all pairs of adjacent datasets that differ by $d_j$, denoted as $\smash{D\stackrel{d_j}{\sim}D_{-j}}$. For clarity, we refer to this variant as \textit{record-specific} adjacent datasets and we have the relationship $\{(D,D_{-j})\} \subset \{(D,D')\}$.

\begin{definition}[$(\mathcal{E},\delta)$-Personalized Differential Privacy \cite{jorgensen2015conservative}] \label{def:individual_dp}
Given a dataset $D$ with each record $d_j\in D$ corresponding to a specific privacy budget $\varepsilon_j>0$. Let $\mathcal{E}=\{\varepsilon_j\}_{j\in[N]}$. A randomized algorithm $\mathcal{A}:\mathbb{D}\rightarrow\mathbb{O}$ satisfies $(\mathcal{E},\delta)$-personalized differential privacy (PDP) if it guarantees $(\varepsilon_j,\delta)$-DP w.r.t. the specific record $d_j$. That is, for any pair of record-specific adjacent datasets $D, D_{-j}\in \mathbb{D}$ and any subsets of output $o\subseteq \mathbb{O}$, it holds that 
$$\Pr[\mathcal{A}(D) \in o] \leq e^{\varepsilon_j} \Pr[\mathcal{A}(D_{-j}) \in o] + \delta.$$
\end{definition}

\begin{remark}
Although $\delta$ is also an important DP parameter and technically its value could be randomly specified like $\varepsilon$, we assume all records share a common $\delta$ with a small, positive default value in this paper based on the following two considerations.
\begin{itemize}[leftmargin=*]
\item On the one hand, $\delta$ is commonly taken to be ``sub-polynomially small'', that is, a rule-of-thumb is that it should be much smaller than the inverse of any polynomial in the size of the dataset~\cite{dwork2014algorithmic,dwork2016concentrated}. Since individuals may not have access to the complete dataset or information about its size, it becomes difficult for them to properly set a value for $\delta$ that meets the desired privacy budgets.
\item On the other hand, the choices of $\varepsilon$ and the choices of $\delta$ are statistically independent, that is, for two different records $d_1,d_2\in D$, if $\varepsilon_1\geq\varepsilon_2$, it is not necessarily always $\delta_1\geq\delta_2$ (and vice versa). We argue this issue is complicated and leave it as an open problem.
\end{itemize}
\end{remark}

\noindent\textbf{The Sample Mechanism.} Building upon the findings of privacy amplification by random sampling~\cite{li2012sampling,abadi2016deep,balle2018privacy,wang2019subsampled,zhu2019poission}, Jorgensen et al. \cite{jorgensen2015conservative} proposed the Sample mechanism. It achieves ($\mathcal{E},0$)-PDP by applying an arbitrary mechanism that satisfies $\varepsilon$-DP on a subset of data records which is obtained by a \textit{non-uniform Poisson sampling} procedure. Our work is inspired by this idea but encounters greater challenges due to the utilization of RDP, detailed in Section \ref{method:challenges}. 

% With the non-uniform per-record sampling probabilities of all records $\{q_i\}_{i\in[N]}$, a subset of $D$ can be selected via Poisson sampling. Combined with the Gaussian mechanism, it can be proved that this PoiSGM satisfies personalized differential privacy.

\begin{definition}[Poisson Sampling \cite{zhu2019poission}] \label{def:poisson_sampling}
Given a dataset $D$ with size $N$ and a set of per-record sampling probabilities $\vect{q} = \{q_i|q_i\in [0,1], i\in [N]\}$, the Poisson sampling procedure outputs a subset $\{d_i |\beta_i =1,i\in [N]\}$ by sampling a Bernoulli random variable $\beta_i\sim Ber(q_i)$ independently. Here $\beta_i\in\{0,1\}$ denotes an indicator that depicts each individual’s participation in the dataset.
\end{definition}

\begin{definition}[Poisson-Sampled Gaussian (PoiSG) mechanism] \label{def:poisgm}
Let $D\in \mathbb{D}$ be an input dataset and $\vect{q}=\{q_1,\dots,q_N\}$ denote the set of sampling probabilities of each record $d_i\in D$. Consider a function $f:\mathbb{D}\rightarrow \mathbb{O}$ with $\ell_2$-sensitivity $L$, then the Poisson-Sampled Gaussian (PoiSG) mechanism is defined as:
$$PoiSG_{\vect{q},\sigma}(D) \triangleq f(S) + \zeta, \quad\zeta\sim \mathcal{N}(0,\sigma^2L^2),$$
where each element $d_i\in S\subseteq D$ is selected via Poisson sampling, and $\mathcal{N}(0,\sigma^2L^2)$ is a Gaussian distribution with standard deviation $\sigma L$. Note that we assume $L=1$ throughout the rest of this paper.
\end{definition}

\subsection{R\'enyi Differential Privacy}

R\'enyi differential privacy (RDP) utilizes the asymmetric measure of R\'enyi divergence to quantify the privacy guarantee. Note that with a controlling parameter $\alpha\neq 1$, the R\'enyi divergence of order $\alpha$ from distribution $Q$ to $P$ is:
\begin{equation*}\begin{split}
D_\alpha(P\|Q) \triangleq \frac{1}{\alpha-1}\log \mathbb{E}_{o\sim Q}\left[\left(\frac{P(o)}{Q(o)}\right)^\alpha\right].
\end{split}\end{equation*}
Let $P=\mathcal{A}(D)$ and $Q=\mathcal{A}(D')$, then ($\alpha,\rho$)-RDP is achieved by simultaneously bounding the R\'enyi divergence of two directions, denoted by $D_\alpha^{\leftrightarrow}(P\|Q)\triangleq \max\{D_\alpha(P\|Q), D_\alpha(Q\|P)\}$.
\begin{definition} [($\alpha,\rho$)-R\'enyi Differential Privacy \cite{mironov2017renyi}] \label{def:renyi_dp}
A randomized mechanism $\mathcal{A}$ satisfies ($\alpha,\rho$)-RDP with order $\alpha\in (1,\infty)$ if
for any pair of adjacent datasets $D,D'\in \mathbb{D}$, 
it holds that
\begin{equation}\begin{split}
D_\alpha^{\leftrightarrow}(\mathcal{A}(D)\|\mathcal{A}(D')) \leq \rho.
\end{split}\end{equation}
\end{definition}

Different from the traditional $(\varepsilon,\delta)$-DP, which measures privacy leakage by utilizing the max divergence of two output distributions, RDP allows for a continuous spectrum of privacy measures. More specifically, as $\alpha\rightarrow \infty$, $D_\infty(\mathcal{A}(D)\|\mathcal{A}(D'))$ is equal to the max divergence \cite{dwork2014algorithmic}; and $\lim_{\alpha\rightarrow 1}D_\alpha(\mathcal{A}(D)\|\mathcal{A}(D'))$ can be verified to be equal to the expected value of the privacy cost random variable $c(o;\mathcal{A},D,D')\triangleq \ln\frac{\Pr[\mathcal{A}(D) \in o]}{\Pr[\mathcal{A}(D') \in o]}$ \cite{mironov2017renyi}. This characteristic enables RDP to provide a sharper privacy quantification and become one of the most popular privacy analysis tools, especially adept at handling composite mechanisms like differentially private stochastic gradient descent (DP-SGD) \cite{abadi2016deep}. 

We provide the following useful lemmas which are important primitives for the design of our FL algorithm and privacy analysis.

\begin{lemma}[Transition from RDP to DP \cite{mironov2017renyi}]\label{lemma:rdp_dp_transition}
If $\mathcal{A}$ is an ($\alpha,\rho$)-RDP mechanism, it also satisfies ($\rho+\frac{\log{1/\delta}}{\alpha-1},\delta$)-DP for any $0<\delta<1$.
\end{lemma}

\begin{lemma}[Adaptive sequential composition \cite{mironov2017renyi}]\label{lemma:adaptive_comp}
If $\mathcal{A}_1: \mathbb{D} \rightarrow \mathbb{O}_1$ is ($\alpha, \rho_1$)-RDP and $\mathcal{A}_2: \mathbb{D}\times\mathbb{O}_1 \rightarrow \mathbb{O}_2$ is ($\alpha, \rho_2$)-RDP, then the composed mechanism $\mathcal{A} \triangleq \mathcal{A}_1 \circ \mathcal{A}_2: \mathbb{D} \rightarrow \mathbb{O}_1\times\mathbb{O}_2$ satisfies ($\alpha, \rho_1 + \rho_2$)-RDP.
\end{lemma}

\begin{lemma}[Post-processing \cite{mironov2017renyi}]\label{lemma:rdp_post_processing}
If $\mathcal{A}$ is ($\alpha,\rho$)-RDP and $\mathcal{F}:\mathbb{O} \rightarrow \mathbb{O}'$ is an arbitrary data-independent randomized mapping, then $\mathcal{F} \circ \mathcal{A}$ is ($\alpha,\rho$)-RDP.
\end{lemma}

\begin{lemma}[Privacy amplification via (uniform) Poisson sampling for Gaussian mechanism \cite{mironov2019renyi,zhu2019poission}]\label{lemma:rdp_amplification}
Consider a $\text{PoiSG}$ mechanism and a uniform sampling probability $q$ among all records. For all pairs of adjacent datasets $D, D'$ and integer $\alpha> 1$, we have\footnote{Note that \cite{mironov2019renyi} and \cite{zhu2019poission} demonstrated similar RDP upper bounds for the PoiSG mechanism. The presented Lemma \ref{lemma:rdp_post_processing} is mainly rooted in the findings of \cite{zhu2019poission}. Specifically, when we work with the Gaussian mechanism, Proposition 10 \cite{zhu2019poission} implies Theorem 8 holds and the lower bound in Theorem 6 is a tighter RDP upper bound for the PoiSG mechanism.}
\begin{equation}\begin{split}
\rho_{\text{PoiSG}}(\alpha,q) \leq &\frac{1}{\alpha-1}\log\left\{
(1-q)^{\alpha-1}(\alpha q-q+1) \right.\\ &\left.+\sum_{\ell=2}^\alpha \binom{\alpha}{\ell} (1-q)^{\alpha-\ell}q^\ell e^{(\ell-1)\rho(\ell)}
\right\}.
\end{split}\end{equation}
Note that $\rho(\alpha)=\frac{\alpha}{2\sigma^2}$ for any $\alpha>1$.
\end{lemma}

\noindent\textbf{Privacy Bounds Visualization.} 
The privacy guarantee under RDP can be depicted as a curve of R\'enyi divergence, aka., the \textit{RDP budget curve}, over the continuous range of $\alpha$ values \cite{mironov2017renyi}. For a clear understanding, we visualize the RDP budget curve of the PoiSG mechanism with uniform sampling probability $q$ in Figure \ref{fig:rdp_dp_curves_sgd} (left). When $q=1$, the PoiSG mechanism is reduced to a standard Gaussian mechanism whose RDP budget curve is a straight line \cite{mironov2017renyi}. For the RDP budget curves with $q<1$, there exists a phase transition that happens around $\alpha q e^{\rho(\alpha)}\approx q^{-1}$ \cite{zhu2019poission}. As $q$ gets larger, this transition tends to appear earlier and get more indistinct. 

Based on Lemma \ref{lemma:rdp_dp_transition}, we can obtain the corresponding \textit{DP budget curve} given a desired $\delta$ and then find the smallest $\varepsilon$ by solving the optimization problem below \cite{abadi2016deep}:
\begin{small}\begin{align}
&\varepsilon^* \triangleq \min_{\alpha}\left\{ \rho + \frac{\log(1/\delta)}{\alpha-1}\right\}.
\end{align}\end{small}
Corollary 38 in \cite{wang2019subsampled} proves the unimodality/quasi-convexity of this optimization problem. Figure \ref{fig:rdp_dp_curves_sgd} (right) demonstrates the existence of an optimal order $\alpha^*$, corresponding to the minimum $\varepsilon^*$.

\begin{remark} 
Instead of exploring an infinite range of real numbers $\alpha\in (1,\infty)$, practitioners often opt to predefine a finite collection of RDP orders to effectively capture the minimum $\varepsilon^*$. This trick has been implemented in leading DP libraries such as Opacus\footnote{https://github.com/pytorch/opacus}, Tensorflow Privacy\footnote{https://github.com/tensorflow/privacy}, etc.
\end{remark}

\begin{figure}[t]\centering
\includegraphics[width=0.86\linewidth]{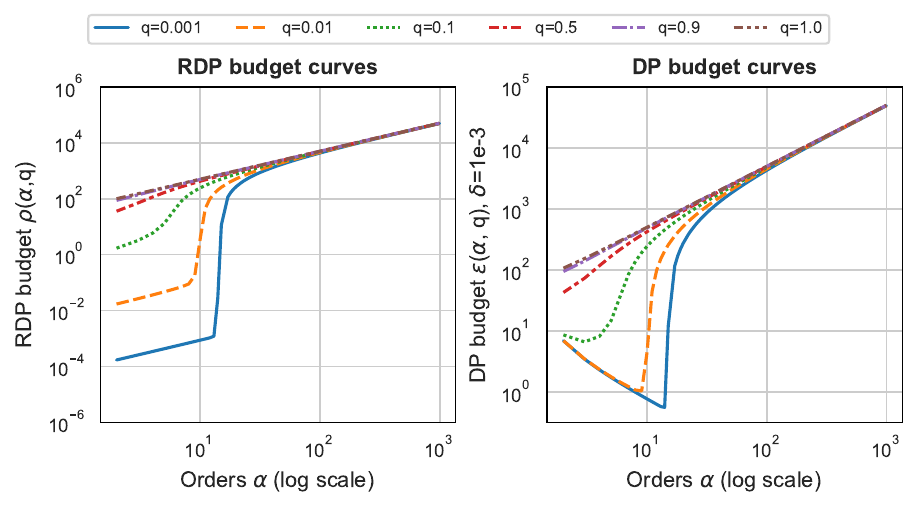}
\caption{The RDP and DP budget curves w.r.t. order $\alpha$ and sampling probability $q$ of a sequential combination of $T$=100 PoiSG mechanisms with noise multiplier $\sigma$=1.0 and $\delta$=1e-3.}
\Description{The RDP and DP budget curves w.r.t. order $\alpha$ and sampling probability $q$ of a sequential combination of $T$=100 PoiSG mechanisms with noise multiplier $\sigma$=1.0 and $\delta$=1e-3.}
\label{fig:rdp_dp_curves_sgd}
\end{figure}

\section{Federated Learning with Record-level Personalized DP}\label{sec:problem_formulation}

We target the typical supervised FL task with a central server and a set of $M$ clients $\mathcal{C}=\{C_1,\dots, C_M\}$. Consider each client $C_i\in \mathcal{C}$ holds a private training dataset $D_i=\{d_{i,1},\dots,d_{i,N_i}\}$. Each record $d_{i,j}\in D_i$ is associated with a privacy budget $\varepsilon_{i,j}>0$, which reflects the privacy preference of the record's owner. Our goal is to learn a globally shared model with parameters $\x\in \mathbb{R}^d$ by solving the following empirical risk problem
\begin{small}\begin{align}\label{eq:erm_obj}
&\min_{\x\in\mathbb{R}^d}\left\{\mathcal{L}(\x) \triangleq \frac{1}{M}\sum_{i=1}^M \mathcal{L}_i(\x; D_i)\right\}, \\
&\text{where } \mathcal{L}_i(\x; D_i) \triangleq \frac{1}{N_i} \sum_{j=1}^{N_i} l(\x, d_{i,j}),
\end{align}\end{small}
with the privacy guarantee of record-level personalized differential privacy (as stated in Definition \ref{def:federated_pdp} below). Here $l(\cdot)$ denotes the loss function used for local optimization.

\begin{definition}[Federated Learning with Record-level Personalized Differential Privacy] \label{def:federated_pdp}
Given $\delta\geq 0$. Let $\mathcal{D}=\bigcup_{i=1}^M D_i$ and $\mathcal{E}=\bigcup_{i=1}^M \{\varepsilon_{i,j}\}_{j\in[N_i]}$. A randomized FL algorithm $\mathcal{A}_{FL}:\mathbb{D}\rightarrow\mathbb{O}$ satisfies $(\mathcal{E},\delta)$-record-level personalized differential privacy (rPDP) if it guarantees $(\varepsilon_{i,j},\delta)$-DP w.r.t. the specific record $d_{i,j}$, i.e., for any pair of record-specific adjacent datasets $\smash{\mathcal{D}\stackrel{d_{i,j}}{\sim}\mathcal{D}_{-i,j}}$ and any subsets of output $o\subseteq \mathbb{O}$, it holds that 
\begin{equation}
\Pr[\mathcal{A}_{FL}(\mathcal{D}) \in o] \leq e^{\varepsilon_{i,j}} \Pr[\mathcal{A}_{FL}(\mathcal{D}_{-i,j}) \in o] + \delta,
\end{equation}
where $\mathcal{D}_{-i,j} \triangleq D_{i,-j}\cup \{\cup_{m\neq j} D_m\}$ and $D_{i,-j} \triangleq D_i\setminus \{d_{i,j}\}$.
\end{definition}

\begin{figure}[t]\centering
\includegraphics[width=\linewidth]{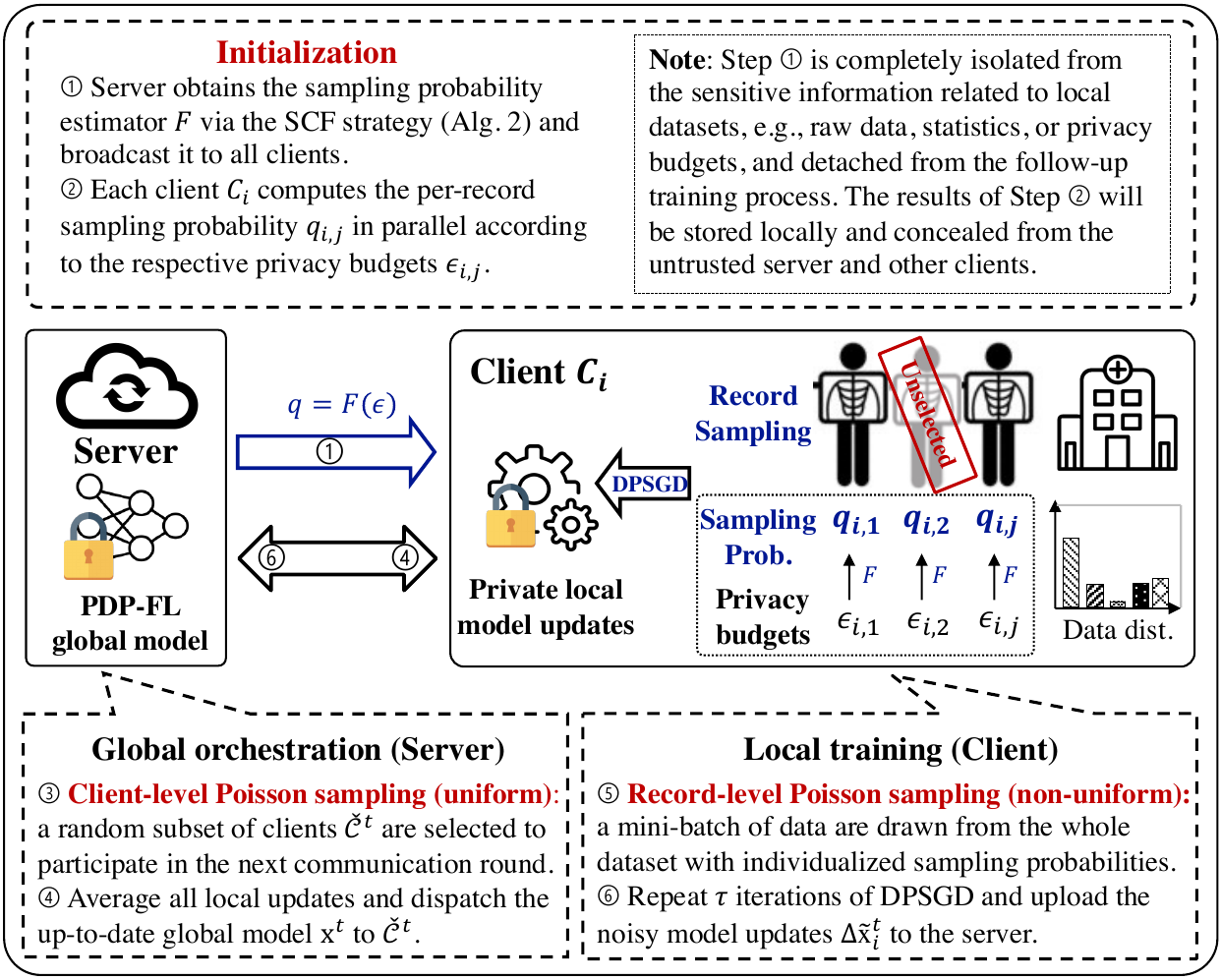}
\caption{A step-by-step illustration of the rPDP-FL algorithm.}
\Description{A step-by-step illustration of the rPDP-FL algorithm.}
\label{fig:rpdpfl_framework}
\end{figure}

\noindent\textbf{FedAvg and Two-stage Sampling Scheme.} The most fundamental approach for solving the non-private optimization problem in Eq. (\ref{eq:erm_obj}) is federated averaging (FedAvg) \cite{mcmahan2017communication}. Despite numerous improved methods being proposed to tackle FedAvg's limitations like data heterogeneity or communication efficiency \cite{karimireddy2020scaffold,li2020federated, hanzely2020lower,li2021ditto}, most of them still adhere to a \textit{two-stage sampling} scheme (i.e., outer client sampling followed by the inner record sampling) together with SGD-based learning paradigm. Given the primary aim of this work is the establishment of a record-level personalized privacy protection mechanism, we choose FedAvg as the backbone of our FL framework and adopt DP-SGD \cite{abadi2016deep} during each client's local training process to achieve record-level protection. We expect our proposed rPDP-FL to be extendable to work with other two-stage sampling FL methods listed above. 

\subsection{Solution Overview}
We employ a two-stage \textit{hybrid} sampling scheme within the FedAvg algorithm to obtain a global model using clients' local datasets while ensuring diverse individual privacy preferences. This innovative framework for private FL, termed \textit{rPDP-FL}, differs from the classic approach in three key aspects:

\begin{itemize}[leftmargin=*]
\item \textbf{Initialization}: 
Each client allocates a customized sampling probability $q_{i,j}$ to every record in its local dataset, tailored to the record's specific privacy budget.
\item \textbf{Stage 1: Client-level Poisson sampling (uniform)}: at the beginning of round $t\in[T]$, the central server selects a random subset of clients $\check{\mathcal{C}}^t$ via Poisson sampling with \textit{uniform} per-client sampling probability $\lambda\in[0,1]$ and dispatches the up-to-date global model $\x^t$ to these selected clients.
\item \textbf{Stage 2: Record-level Poisson sampling (non-uniform)}: each client selected in the above stage performs a certain number of DP-SGD iterations locally and independently and uploads the model updates to the central server. During each iteration, the mini-batches are drawn from the whole local dataset via Poisson sampling with non-uniform per-record sampling probabilities.
\end{itemize}

It's worth highlighting that rPDP-FL solely alters the sampling processes (except for the initialization step) and remains detached from the intricacies associated with the learning process. This feature enables its broader applicability to any non-private FL frameworks that incorporate a two-stage sampling process, \textcolor{black}{as illustrated in Figure \ref{fig:rpdpfl_framework}.} The pseudocode is presented in Algorithm \ref{alg:pers_dp_fedavg_pesudo} and the complete version will be shown in Algorithm \ref{alg:pers_dp_fedavg}. 

\begin{algorithm}[t] \small
\caption{Record-level Personalized Differentially Private Federated Learning (rPDP-FL, Pseudocode)}
\label{alg:pers_dp_fedavg_pesudo}
\DontPrintSemicolon
\SetKwInOut{Input}{input}
\SetKwInOut{Parameter}{parameter}
\SetKwInOut{Output}{output}
\Input{$M$ clients with their local datasets $(D_1,\dots,D_M)$; the total communication round $T$ and the local SGD step $\tau$; the client-level sampling probability $\lambda$.}
\color{blue}\tcp{Initialization}\color{black}
\ForEach(\emph{in parallel}){client $C_i \in \mathcal{C}$ }{
    $\{q_{i,j}\}_{j\in[|D_i|]}\leftarrow$ (pre-computation of sampling probabilities for all records) \;
}
\For{$t\in[T]$ }{
    \color{blue}\tcp{Client-level Poisson sampling with the uniform sampling probability $\lambda$}\color{black}
    $\check{\mathcal{C}}^t\leftarrow$ (a random subset drawn from $[M]$) \;
    \ForEach(\emph{in parallel}){client $C_i \in \check{\mathcal{C}}^t$ }{
        \For{$r\in[\tau]$ }{
            \color{blue}\tcp{Record-level Poisson sampling with non-uniform sampling probabilities $\{q_{i,j}\}_{j\in[|D_i|]}$}\color{black}
            $S^r\leftarrow$ (a random mini-batch drawn from $D_i$) \;
            \tcp{Differentially private SGD}
        }
    }
    \tcp{The central server averages the collected noisy model updates and obtains the updated global model parameters}
}
\end{algorithm}

\subsection{Challenges} \label{method:challenges}

To offer reasonable personalized privacy guarantees while maintaining the utility of the global model, the development of Algorithm \ref{alg:pers_dp_fedavg_pesudo} faces a \textit{dual} challenge in both theory and practice.
\begin{itemize}[leftmargin=*]
    \item \textbf{The privacy analysis challenge.} From a theoretical perspective, it is essential to establish as ``tight'' upper bounds as possible for the overall privacy cost of each individual to enhance the trade-off between privacy and utility.
    \item \textbf{The parameter estimation challenge.} For practical purposes, an efficient and effective parameter estimation strategy must be adopted to select appropriate hyperparameters for the privacy algorithm, i.e., determining sampling probabilities for all records. 
\end{itemize}

Our research aims to explore the mathematical relationship between privacy cost and their sampling probabilities. In particular, we provide the privacy analysis given the sampling probability of each record, as detailed in Section \ref{sec:privacy}. Furthermore, we outline our approach for deriving the sampling probabilities in accordance with the predetermined individualized privacy budgets in Section \ref{sec:solution}.

% This approach is termed \textit{Simulation-CurveFitting} (SCF), as illustrated in Figure \ref{fig:scf_framework}.

% \begin{figure}[t]\centering
% \includegraphics[width=\linewidth]{figure/scf_framework.pdf}
% \caption{An illustration of Simulation-CurveFitting.}
% \Description{An illustration of Simulation-CurveFitting.}
% \label{fig:scf_framework}
% \end{figure}

\section{Privacy Analysis}\label{sec:privacy}

\subsection{Privacy Objectives and Key Results} We will analyze the upper bound of the accumulative privacy cost for any single record $d_{i,j}$ in Algorithm \ref{alg:pers_dp_fedavg_pesudo}, assuming its sampling probability $q_{i,j}$ is given. This can be broken down into the following three basic nested routines: 
\begin{enumerate}[leftmargin=*]
\item \textit{Local multi-step update}, which can be abstracted as an adaptive combination of $\tau$ PoiSG mechanisms.
\item \textit{Global parameter aggregation}, which can be seen as a multi-phase procedure involving uniform client sampling and data-independent post-processing of the results derived from the local update performed by the chosen clients.
\item \textit{Global multi-round update}, which can be perceived as an adaptive combination of $T$ parameter aggregation mechanisms above.
\end{enumerate}

\noindent\textbf{Symbolic Representations.} Without loss of generality, our focus will be primarily on the first record $d_{1,1}\in D_1$ of client $C_1$. For the sake of conciseness, we use the symbolic representations as follows.
\begin{itemize}[leftmargin=*]
\item $\mathcal{D}\triangleq\bigcup_{i=1}^M D_i$: the federated dataset.
\item $\mathcal{D}, \mathcal{D}_{-1,1}$: the adjacent \textit{federated} datasets concerning a specific data record $d_{1,1}$, i.e., $\smash{\mathcal{D}\stackrel{d_{1,1}}{\sim}\mathcal{D}_{-1,1}}$.
\item $D_1, D_{1,-1}$: the adjacent \textit{local} datasets at client $C_1$ concerning a specific data record $d_{1,1}$, i.e., $\smash{D_1\stackrel{d_{1,1}}{\sim}D_{1,-1}}$.
\item $\text{CSamp}_\lambda(D_1,\dots, D_M)$: client sampling using uniform Poisson sampling, where $\lambda\in (0,1]$ denotes the sampling probabilities for all clients.
\item $\text{RSamp}_{\mathbf{q}_1}(D_1)$: record sampling at client $C_1$ using non-uniform Poisson sampling, where $\mathbf{q}_1=\{q_{1,1},\dots,q_{1,|D_1|}\}$ denotes the set of sampling probabilities of each record $d_{1,j}\in D_1$. 
\item $\mathcal{A}_G(\cdot) \triangleq f(\cdot) + \mathcal{N}(0,\sigma^2\mathbb{I})$: the Gaussian mechanism satisfying $(\alpha,\rho_G)$-RDP, where $\rho_G(\alpha)\triangleq \frac{\alpha L^2}{2\sigma^2}$ and $L$ is the the sensitivity of function $f$ \cite{mironov2017renyi}. For simplicity, we assume that $L=1$ through the rest of this section.
\item $\mathcal{A}(\cdot) \triangleq \mathcal{A}_G(\text{RSamp}_{\mathbf{q}_1}(\cdot))$: the PoiSG mechanism.
\item $\mathcal{A}_{in}(\cdot) \triangleq \mathcal{A}^{\otimes \tau} = (\mathcal{A}_1(\cdot),\mathcal{A}_2(\mathcal{A}_1(\cdot),\cdot),\dots,\mathcal{A}_\tau(\mathcal{A}_1(\cdot),\dots,\cdot))$: local multi-step update.
\item $\mathcal{A}_{out}(\cdot) \triangleq \text{Avg}(\text{CSamp}_\lambda\circ\mathcal{A}_{in})$: global parameter aggregation.
\item $\mathcal{A}_{FL}(\cdot) \triangleq \mathcal{A}_{out}(\cdot)^{\otimes t}$: global multi-round update.
\end{itemize}

To enjoy the strength of tight privacy accounting offered by the RDP privacy analysis framework, we need to overcome the incompatibility challenge that all existing PDP techniques fail to provide tight privacy analysis under the RDP framework. In detail, we consider analyzing the RDP bound of $\mathcal{A}_{FL}$ first, and then convert it into the form of a standard DP guarantee by applying Lemma \ref{lemma:rdp_dp_transition}. For example, in order to show that $\mathcal{A}_{FL}$ satisfies ($\varepsilon_{1,1},\delta$)-DP w.r.t. $d_{1,1}$, we need to show that for any pair of adjacent federated datasets $\mathcal{D}, \mathcal{D}_{-1,1}$ and arbitrary output $o$, we have
\begin{small}\begin{align}
&D(\mathcal{A}_{FL}(\mathcal{D}) \| \mathcal{A}_{FL}(\mathcal{D}_{-1,1})) \leq \rho_{FL}, \nonumber\\
&\text{s.t.}\quad \min_{\alpha>1}\left\{ \rho_{FL} + \frac{\log(1/\delta)}{\alpha-1}\right\} \leq \varepsilon_{1,1}. \label{eq:rdp_dp_trans}
\end{align}\end{small}

\noindent\textbf{Privacy Objectives.} As previously discussed in Section \ref{sec:related_work}, FL scenarios typically account for two distinct types of potential adversaries. Consequently, the objectives of privacy analysis can be categorized into the following.
% the \textit{honest-but-curious} central server (which has access to both intermediate model updates and final global model), and the \textit{untrusted} third parties or clients (which have access to the final global model).
% \footnote{Even if each client uploads its local parameter update $\Delta \x_i$ to the server, the server can still obtain local model parameters. The resulting privacy cost is the same in both cases.}.
\begin{enumerate}[leftmargin=*]
\item \textit{Type I privacy analysis} against the honest-but-curious server (which has access to the intermediate model updates): given the local model parameter $\x_1^t\sim\mathcal{A}_{in}(D_1)$ uploaded by Client $C_1$, the individual RDP privacy bound of $\mathcal{A}_{FL}$ for record $d_{1,1}$ is 
\begin{small}\begin{align*}
D_\alpha(\mathcal{A}_{FL}(\mathcal{D}) &\| \mathcal{A}_{FL}(\mathcal{D}_{-1,1})) \\
&= \lambda T\cdot D_\alpha(\mathcal{A}_{in}(D_1) \| \mathcal{A}_{in}(D_{1,-1})) \leq \rho_{I}.
\end{align*}\end{small}
\item \textit{Type II privacy analysis} against untrusted clients or third parties (which have access to the intermediate or final global model): for any global model parameter $\x^t\sim\mathcal{A}_{FL}(\mathcal{D})$, the individual RDP privacy bound of $\mathcal{A}_{FL}$ for record $d_{1,1}$ is 
\begin{small}\begin{align*}
D_\alpha(\mathcal{A}_{FL}(\mathcal{D}) &\| \mathcal{A}_{FL}(\mathcal{D}_{-1,1})) \\
&= T\cdot D_\alpha(\mathcal{A}_{out}(\mathcal{D}) \| \mathcal{A}_{out}(\mathcal{D}_{-1,1})) \leq \rho_{II}.
\end{align*}\end{small}
\end{enumerate}

\noindent\textbf{Key Results.} The key results are presented below. The detailed proofs will be presented in the next subsection. %The main objective of the following section is to theoretically analyze the upper bounds of the accumulative privacy cost for data record $d_{1,1}$ under the two scenarios mentioned above. To illustrate the overall analytical approach, we first use a toy example to explain the enhanced privacy effects at different stages and extend the conclusion to more general scenarios.

\begin{lemma}\label{lemma:post_process}
Suppose that $\check{\mathcal{C}}^t$ is a subset of clients selected at round $t\in [T]$. The simple average operation $\text{Avg}(\cdot)$ over all outputs $o_i \sim \mathcal{A}_{in}(D_i)$, where $C_i\in \check{\mathcal{C}}^t$, will not incur any extra privacy cost to \textcolor{black}{all records $d_{i,j}\in\mathcal{D}$}. 
\end{lemma}
\begin{proof} The proof follows from the fact that the RDP guarantee is preserved under post-processing, as shown in Lemma \ref{lemma:rdp_post_processing}.
\end{proof}

\begin{lemma}[Enhanced Privacy Amplification by Two-Stage Hybrid Sampling]
\label{lemma:enhan_indiv_priv_amp}
Assume the sampling probability for any clients is $\lambda\in(0,1]$, and the sampling probability for data record $d_{i,j}\in D_i$ is $q_{i,j}\in (0,1]$. If a random algorithm $\mathcal{A}_{in}(D_i,\x^{t-1})$ satisfies $(\alpha, \rho_{i,j}^{\tau})$-RDP w.r.t $d_{i,j}$, then the algorithm $\mathcal{A}_{out}(\mathcal{D})$ satisfies $(\alpha, \rho_{i,j}^{\tau,\lambda})$-RDP w.r.t. $d_{i,j}$, where
\begin{small}\begin{align*}
\rho_{i,j}^{\tau,\lambda}(\alpha,q_{i,j}) 
&\leq \frac{1}{\alpha-1} \ln\left\{1-\lambda+\lambda e^{(\alpha-1)\rho_{i,j}^{\tau}(\alpha,q_{i,j})} \right\}, \text{~and~}
\end{align*}\end{small}
\begin{small}\begin{align*}
\rho_{i,j}^{\tau}(\alpha,q_{i,j}) \leq \frac{\tau}{\alpha-1}\ln&
\left\{(1-q_{i,j})^{\alpha-1}(\alpha q_{i,j}-q_{i,j}+1) \right.\nonumber\\ 
&\left. +\sum_{\ell=2}^\alpha \binom{\alpha}{\ell} (1-q_{i,j})^{\alpha-\ell}q_{i,j}^\ell e^{(\ell-1)\rho_G(\ell)}
\right\}.
\end{align*}\end{small}
\end{lemma}

\begin{theorem}[Individual Privacy Analysis in Federated Learning]\label{theorem:indiv_priv_analysis}
For any $\delta\in (0,1)$, the random algorithm $\mathcal{A}_{FL}(\mathcal{D})$ satisfies $(\hat{\varepsilon}_{i,j}^*, \delta)$-DP w.r.t. a specific record $d_{i,j}\in \mathcal{D}$, where
\begin{small}\begin{align}\label{eq:final_epsilon}
\hat{\varepsilon}_{i,j}^*\triangleq \min_{\alpha} \left(\rho_{FL}(\alpha,q_{i,j})+\frac{\ln(1/\delta)}{\alpha-1}\right).
\end{align}\end{small}
Note that:
(1) for untrusted clients or third parties, $\rho_{FL}(\alpha,q_{i,j}) \triangleq T\rho_{i,j}^{\tau,\lambda}(\alpha,q_{i,j})$;
% \begin{small}\begin{align*}
% \rho_{FL}(\alpha,q_{i,j}) &\triangleq T\rho_{i,j}^{\tau,\lambda}(\alpha,q_{i,j}) \\
% &\leq \frac{T}{\alpha-1} \ln\left\{1-\lambda + \lambda \big((1-q_{i,j})^{\alpha-1}(\alpha q_{i,j}-q_{i,j}+1) \big.\right.\\
% &\quad \left.\big.+\sum_{\ell=2}^\alpha \binom{\alpha}{\ell} (1-q_{i,j})^{\alpha-\ell}q_{i,j}^\ell e^{(\ell-1)\rho_G(\ell)}\big)^\tau \right\}
% \end{align*}\end{small}
(2) for the honest-but-curious server, $\rho_{FL}(\alpha,q_{i,j}) \triangleq \lambda T\rho_{i,j}^{\tau}(\alpha,q_{i,j})$.
% \begin{small}\begin{align*}
% \rho_{FL}(\alpha,q_{i,j}) \triangleq \lambda T\rho_{i,j}^{\tau}(\alpha,q_{i,j}).
% \end{align*}\end{small}
% and the expression of $\rho_{i,j}^{\tau}(\alpha,q_{i,j})$ is shown in Eq. \ref{eq:rho_i_j_tau}.
\end{theorem}

\subsection{Detailed Proofs}

We first use a special case to explain the enhanced privacy effects at different stages and extend the conclusion to more general scenarios. Consider that a server collaborates with two clients, $C_1$ and $C_2$, to collectively train an FL model. Here $C_2$ is assumed to be an adversary and aims to infer whether $d_{1,1}$ is contained in $D_1$.

Considering the sequential composition of  RDP  as  in Lemma \ref{lemma:adaptive_comp}, our major objective will be analyzing the increment of the individual RDP parameter between two successive rounds, that is,
\begin{align*}
&D(\mathcal{A}_{out}(\mathcal{D}) \| \mathcal{A}_{out}(\mathcal{D}_{-1,1})) \\
&\triangleq \frac{1}{\alpha-1} \log \mathbb{E}_{o\sim \mathcal{A}_{out}(\mathcal{D}_{-1,1})} \left[\left(\frac{\Pr[\mathcal{A}_{out}(\mathcal{D}) \in o]}{\Pr[\mathcal{A}_{out}(\mathcal{D}_{-1,1}) \in o]}\right)^\alpha\right].
\end{align*}

\subsubsection{Local Multi-step Update}
According to Fact \ref{fact1} below which aligns with the decentralized nature of FL, the privacy analysis of the local multi-step update process essentially follows the existing theoretical results based on the typical DP-SGD algorithm. The only difference is that now we need to characterize the privacy cost for each record since the sampling probabilities of the records are different from each other.

\begin{fact}\label{fact1}
Once the initial model parameters $\x^{t-1}$ are fixed at the beginning of round $t\in[T]$, each client performs local update independently, i.e., the distribution of the output $\x_i^t \sim \mathcal{A}_{in}(\x^{t-1}, D_i)$ ($i=2,\dots,M$) is independent of any data records in $D_1$. 
\end{fact}

For the local dataset $D_1$ with a size of $N$, let $\vect{s} = (s_1,\dots,s_N) \subseteq \{0,1\}^N$ be the indicator vector of the record sampling outcome, i.e., $s_j=1$ if $d_{1,j}$ is selected\footnote{For notational convenience, we suppress the dependence on the client identifier $i$.}. It is evident the probability that $\vect{s}$ appears is $p_\vect{s}=\prod_{j=1}^N (q_j)^{s_j}(1-q_j)^{1-s_j}$ and the total number of possible values of $\vect{s}$ is $2^N$. For example, if $N=3$ and $\vect{s} = [1,0,1]$, then $p_\vect{s} = q_1(1-q_2)q_3$ and total number of possible values of $\vect{s}$ is 8. Then for any pair of adjacent local datasets $D_1, D_{1,-1}$ and any subsets of output $o\subseteq \mathbb{O}$, the output distributions of a single DP-SGD step can be represented as:
\begin{small}\begin{align*}
\Pr[\mathcal{A}(D_1)\in o] 
&= \sum_{\vect{s}} p_\vect{s} \Pr[\mathcal{A}_G(\vect{s})\in o] \\
&= (1-q_1)\sum_{\vect{s}:s_1=0} p_\vect{s} \Pr[\mathcal{A}_G(\vect{s})\in o|s_1=0]
\\ 
&\quad + q_1\sum_{\vect{s}:s_1=1} p_\vect{s} \Pr[\mathcal{A}_G(\vect{s})\in o|s_1=1] \\
\Pr[\mathcal{A}(D_{1,-1})\in o] &= \sum_{\vect{s}:s_1=0} p_{\vect s} \Pr[\mathcal{A}_G(\vect{s})\in o|s_1=0]
\end{align*}\end{small}

As the local multi-step update process $\mathcal{A}_{in}$ can be viewed as a $\tau$-fold adaptive composition of a PoiSG mechanism, we have the following Lemma \ref{lemma:local_dpsgd} by directly applying the existing RDP composition and amplification results as shown in Lemma \ref{lemma:adaptive_comp} and Lemma \ref{lemma:rdp_amplification}. The distinction lies in the privacy guarantee provided by $\mathcal{A}_{in}$ is specific to individual records, instead of being established on the wider scope of the adjacent datasets $D_1$ and $D_1'$.

\begin{lemma}\label{lemma:local_dpsgd}
For any client $C_i$, if the sampling probability of a specific record $d_{i,j}\in D_i$ is $q_{i,j}\in(0,1]$, then the local multi-step update process $\mathcal{A}_{in}(D_i)$ satisfies $(\alpha,\rho_{i,j}^\tau)$-RDP w.r.t. $d_{i,j}$, where
\begin{equation}\label{ineq:poi_sgm_indiv}
\rho_{i,j}^\tau \triangleq\tau\cdot\rho_{PoiSG}(\alpha;D_i,D_{i,-j},q_{i,j})
\end{equation}
\end{lemma}

% \begin{remark}
% If we consider the global privacy guarantee of $\mathcal{A}_{in}$ in terms of the whole adjacent datasets $D_1, D_1'$, things are different.
% \end{remark}

\subsubsection{Global Parameter Aggregation}
We consider the output distribution of $\mathcal{A}_{out}$ on the adjacent federated datasets $\mathcal{D},\mathcal{D}_{-1,1}$ in the context of the above special case. Let $P_i \triangleq \Pr[\mathcal{A}_{in}(D_i) \in o_i]$ ($i=1,2$) and $P_1' \triangleq \Pr[\mathcal{A}_{in}(D_{1,-1}) \in o_1]$. 
It can be observed that for the federated dataset $\mathcal{D}$, the underlying distribution $\Phi\triangleq \Pr[\mathcal{A}_{out}(\mathcal{D})\in o]$ can be represented as
\begin{itemize}[leftmargin=*]
\item a mixture of $P_1$ and $P_2$, denoted as $H_{11}$, if both $C_1$ and $C_2$ are selected;
\item the same as $P_1$, denoted as $H_{10}$, if only $C_1$ is selected;
\item the same as $P_2$, denoted as $H_{01}$, if only $C_2$ is selected;
\item independent of both, denoted as $H_{00}$, if neither $C_1$ nor $C_2$ is selected.
\end{itemize}
Similarly, the distribution $\Psi\triangleq \Pr[\mathcal{A}_{out}(\mathcal{D}_{-1,1})\in o]$ will be
\begin{itemize}[leftmargin=*]
\item a mixture of $P_1'$ and $P_2$, denoted as $H_{11}'$, if both $C_1$ and $C_2$ are selected;
\item the same as $P_1'$, denoted as $H_{10}'$, if only $C_1$ is selected;
\item the same as $H_{01}$, if only $C_2$ is selected;
\item the same as $H_{00}$, if neither $C_1$ nor $C_2$ is selected.
\end{itemize}
\begin{remark}
Note that here we focus on the individual privacy cost for PDP which is measured on all pairs of record-level adjacent datasets w.r.t. the ``target'' record $d_{1,1}$. In the context of traditional (uniform) DP, we cannot simply assume the output from other clients $o_2,\dots,o_M$ are constants when analyzing the impact of an individual record on the worst-case privacy cost, as the output distribution of $\text{Avg}(\cdot)$ is highly dependent on each record in $\mathcal{D}$.
\end{remark}

Let $\omega \subseteq (\omega_1,\dots,\omega_M) \in \{0,1\}^M$ be the indicator vector of the outcome of client sampling and $\omega_i=1$ denotes that Client $C_i$ is selected. Then we have
\begin{small}\begin{align*}
\Phi &= \lambda(1-\lambda)H_{10} + \lambda^2 H_{11}  + (1-\lambda)^2 H_{00} + (1-\lambda)\lambda H_{01}, \\
&= \sum_{\omega:\omega_1=1} p_{\omega}H_\omega + \sum_{\omega:\omega_1=0}p_{\omega}H_\omega.\\
\Psi &= \lambda(1-\lambda) H_{10}' + \lambda^2 H_{11}' + (1-\lambda)^2 H_{00} + (1-\lambda)\lambda H_{01}.\\
&= \sum_{\omega:\omega_1=1} p_{\omega}H_\omega' + \sum_{\omega:\omega_1=0}p_{\omega}H_\omega.
\end{align*}\end{small}

Now we try to bound $\mathbb{E}_{\Psi} \left[ ( \Phi/\Psi )^\alpha\right]$ by means of decomposition and simplification. More specifically, we have
\begin{footnotesize}\begin{align*}
&\mathbb{E}_{\Psi} \left[ ( \Phi/\Psi )^\alpha\right] \overset{(1)}{\leq} \mathbb{E}_{\omega} \mathbb{E}_{\Psi} \left[ ( \Phi/\Psi )^\alpha | \omega\right] \\
%\mathbb{E}_{\omega} \left[\mathbb{E}_{\Phi} \left[ \left( \frac{\Psi(o)}{\Phi(o)} \right)^\alpha \bigg| \omega\right]\right] \\
&= \mathbb{E}_{\omega_2} \left\{ \lambda\mathbb{E}_{H_\omega'} \left[ ( H_\omega/H_\omega')^\alpha | \omega_1=1\right] + (1-\lambda)\mathbb{E}_{H_\omega'} \left[( H_\omega/H_\omega')^\alpha | \omega_1=0\right]\right\} \\
&= \lambda\mathbb{E}_{H_\omega'} \left[ ( H_\omega/H_\omega')^\alpha | \omega_1=1\right] + (1-\lambda) \\
&\overset{(2)}{\leq} \lambda \mathbb{E}_{P_1'} \left[ ( P_1/P_1')^\alpha \right] + (1-\lambda) \\
&= \lambda e^{(\alpha-1)D_{\alpha}(\mathcal{A}_{in}(D_1\|\mathcal{A}_{in}(D_{1,-1}))} + (1-\lambda) \\
&\leq \lambda e^{(\alpha-1)\rho_{1,1}^\tau} + (1-\lambda),
\end{align*}\end{footnotesize}
where the inequality (1) follows from the Jensen’s inequality and Lemma 22 in \cite{wang2019subsampled} which proves bivariate function $f(x,y) = x^\alpha / y^{\alpha-1}$ is jointly convex on $\mathbb{R}_+^2$ for all $\alpha>1$; the inequality (2) follows from Lemma \ref{lemma:post_process} which implies that the individual privacy guarantee for $d_{1,1}$ is immune to post-processing, i.e., 
\begin{small}\begin{align*}
D_\alpha(H_\omega \| H_\omega')\leq D_\alpha(\mathcal{A}_{in}(D_1) \| \mathcal{A}_{in}(D_{1,-1})) \leq \rho_{i,j}^\tau.
\end{align*}\end{small}
The equality holds if each client's sampling probability $\lambda = 1$. Similarly, we can also have
\begin{small}\begin{align*}
\mathbb{E}_{\Phi} \left[ ( \Psi/\Phi )^\alpha\right] 
&\leq \lambda e^{(\alpha-1)D_{\alpha}(\mathcal{A}_{in}(D_{1,-1})\|\mathcal{A}_{in}(D_1))} + (1-\lambda) \\
&\leq \lambda e^{(\alpha-1)\rho_{1,1}^\tau} + (1-\lambda).
\end{align*}\end{small}

Now we have completed the proof of Lemma \ref{lemma:enhan_indiv_priv_amp}. %with $M=2$ clients. We extend the analysis to the general case where the federated learning algorithm involves $M\geq 2$ clients. We defer more details to the Appendix due to space limitations.

% \begin{remark} 
% For the general case, the underlying distribution of $\mathcal{A}_{in}(D_i)$ for $i\in[M]$, i.e., $P_i$, could be very complicated but still be a mixture. Luckily, our analysis is independent of its concrete expressions.
% \end{remark}

\subsubsection{Global Multi-round Update}
Putting all the pieces together, the proof of Theorem \ref{theorem:indiv_priv_analysis} can be further derived by leveraging the adaptive composition theorem of RDP as shown in Lemma \ref{lemma:adaptive_comp}. 

\section{Selecting Sampling Probability}\label{sec:solution}

\begin{figure}[t]\centering
\includegraphics[width=\linewidth]{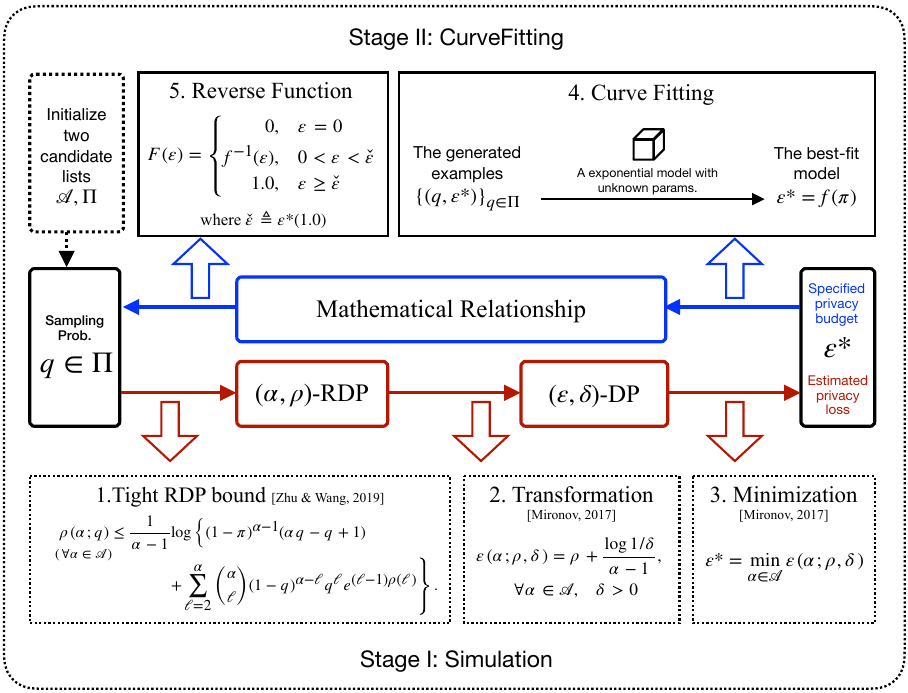}
\caption{An illustration of Simulation-CurveFitting.}
\Description{An illustration of Simulation-CurveFitting.}
\label{fig:scf_framework}
\end{figure}

In this section, we explore how to select a sampling probability for every single record to achieve an estimated privacy cost that closely aligns with the desired privacy budget, on the condition that the other factors (e.g., $T, \tau$, $\sigma$, and $\delta$) remain constant.

\subsection{Simulation-CurveFitting} 
Given the theoretical result in Theorem \ref{theorem:indiv_priv_analysis}, it would be ideal if we could directly derive a sampling probability for each record given its predetermined personalized privacy budget. Yet, it is non-trivial to derive an explicit closed-form expression due to the complexity arising from the optimization process and the highly nonlinear nature of the tight bound $\rho_{FL}$. Existing approaches utilize numerical methods to handle this absence of closed-form issue and approximately obtain the sampling probability, e.g., by binary search algorithm \cite{boenisch2023have}. However, these strategies become computationally demanding when applied to our case. We introduce a new and effective strategy termed \textit{Simulation-CurveFitting} (SCF). As the name suggests, this approach consists of the following two stages, and the specific steps are illustrated in Figure \ref{fig:scf_framework} and outlined in Algorithm \ref{alg:simu_curvefit_strategy}.

\textbf{Stage I: Simulation.} 
We aim to elucidate the relationship between $q$ and $\varepsilon$\footnote{For notational convenience, we suppress the dependence on the record identifier $i,j$.} through a series of simulation experiments. In the beginning, we establish two finite sets of candidate values: the first consists of various sampling probabilities, which we refer to as $\Pi$; the second consists of a sequence of discrete RDP orders denoted as $\mathbb{A}$. For each $q\in \Pi$, we compute the DP budget curves and then find the corresponding minimum value $\varepsilon^*$. For illustration, we show a series of DP budget curves for the rPDP-FL algorithm in Figure \ref{fig:inclusion_prob_est} (left\footnote{This figure closely resembles Figure \ref{fig:rdp_dp_curves_sgd} (right) which depicts results obtained in centralized settings. However, a key distinction is that the minimum values $\varepsilon^*$ across all the curves in Figure \ref{fig:rdp_dp_curves_sgd} (right) are consistently greater than those in Figure \ref{fig:inclusion_prob_est} (left). This highlights the enhanced effect of privacy amplification resulting from the client-level sampling procedure in FL framework.}), corresponding to varying values of $q$. The minimum value on each curve is then used to plot the optimum DP budget w.r.t. the corresponding sampling probability on the right figure.
The pseudocode is illustrated in lines 3-6 of Algorithm \ref{alg:simu_curvefit_strategy}.

\textbf{Stage II: CurveFitting.}
In Figure \ref{fig:inclusion_prob_est} (right), we depict the one-to-one correspondence between the optimum DP budgets and their respective sampling probabilities (represented as a series of dots) across various parameter settings (represented by different colors).
This visualization reveals a compelling observation -- \textit{a mathematical function may potentially model the tight upper bound on the accumulative privacy cost w.r.t. its sampling probability}. Armed with this insight, we employ curve-fitting tools\footnote{SciPy: https://scipy.org/.} to approximate the correlation between $q$ and $\varepsilon$. The best-fit solution obtained  is a simple exponential function in the form
$$\hat{\varepsilon}^* \approx f(q) \triangleq e^{a\cdot q + b} + c,$$
as stated in Line 7 of Algorithm \ref{alg:simu_curvefit_strategy}. Our best-fit function is more concise and elegant than the one given in Eq.(\ref{eq:final_epsilon}) in Theorem \ref{theorem:indiv_priv_analysis}, allowing inverse computation of the sampling probability $q$ given the privacy cost. We refer to the inverse function as \textit{sampling probability estimator}, denoted as $F(\varepsilon)$, which takes a privacy budget $\varepsilon>0$ as input and outputs a valid sampling probability $q\in [0,1]$. In particular, if an input $\varepsilon$ exceeds the optimum DP budget corresponding to $q$=1.0, denoted as $\varepsilon^*(1.0)$, the output probability $q$ is projected to be 1.0. % we enforce the output of $F(\varepsilon)$ being 0 when $\varepsilon$=0, implying records with the ``perfect'' privacy.  
See Line 8 of Algorithm \ref{alg:simu_curvefit_strategy} for more details.

\textbf{Measures for goodness-of-fit.} We utilize the $R^2$ value (also known as the coefficient of determination) to quantify how well the estimated privacy cost by the curve fitting function matches the privacy cost derived from the privacy accounting, ranging from 0 (no correlation) to 1 (perfect positive correlation).
As per the empirical results illustrated in Figure \ref{fig:inclusion_prob_est} (right), the best-fit model exhibits an $R^2$ value exceeding 99\%. This demonstrates strong evidence of the model's ability to derive the sampling probability for each record based on their desired privacy budgets. 

\begin{algorithm}[t] \small
\caption{The Simulation-CurveFitting (SCF) strategy}
\label{alg:simu_curvefit_strategy}
\DontPrintSemicolon
\SetKwInOut{Input}{input}
\SetKwInOut{Parameter}{parameter}
\SetKwInOut{Output}{output}
\Input{%The number of adaptive compositions $k$; 
The noise multiplier $\sigma$, the gradient clipping bound $L$, and the target DP parameter $\delta$.} %the norm clipping bound $L$
\Output{The sampling probability estimator}
\color{blue}\tcp{Initialize two candidate lists of $\alpha,q$} \color{black}
$\mathbb{A}\leftarrow$ a candidate list of RDP order $\alpha\in(1,\infty)$ \;
$\Pi\leftarrow$ a candidate list of sampling probability $q\in[0,1]$\;
\ForEach{$q \in \Pi$}{
    \color{blue}\tcp{Numerical simulation analysis of PoiSGM with sampling probability $q$} \color{black}
    $\rho_{FL}(\alpha,q) \leftarrow$ (the RDP budget curve w.r.t. order $\alpha\in \mathbb{A}$ calculated based on Theorem \ref{theorem:indiv_priv_analysis})\;
    $\varepsilon(\alpha,\delta,q) = \rho_{FL}(\alpha,q) + \frac{\log{1/\delta}}{\alpha-1} \leftarrow$ (the DP budget curve w.r.t. order $\alpha\in \mathbb{A}$ calculated based on Lemma \ref{lemma:rdp_dp_transition}) \;
    $\varepsilon^*(\delta,q) = \min_{\alpha\in\mathbb{A}} \varepsilon(\alpha,\delta,q) \leftarrow$ (the optimum DP budget w.r.t. sampling probability $q$)\;
}
\color{blue}\tcp{Curve fitting}\color{black}
$f(q) \leftarrow$ (the best-fit mathematical model to the generated observations $\{(q, \varepsilon^*)\}_{q\in\Pi}$) \;
\color{blue}\tcp{The sampling probability estimator }\color{black}
% $\check{\varepsilon} \triangleq \varepsilon^*(1.0)\leftarrow$ (the maximum effective value of DP budget) \;
\begin{equation*}\label{eq:inclusion_probability}
F = \left\{
\begin{aligned}
    f^{-1}(\varepsilon), \quad & 0 < \varepsilon < \varepsilon^*(1.0) \\
    1.0, \quad & \varepsilon \geq \varepsilon^*(1.0) \\
\end{aligned}
\right.
\end{equation*}
\Return $F$
\end{algorithm}

\begin{figure}[t]\centering
\includegraphics[width=\linewidth]{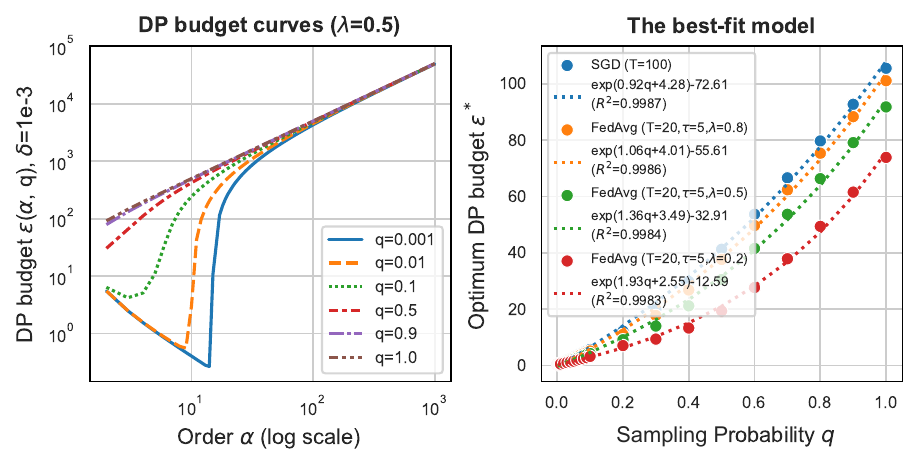}
\caption{The DP budget curves w.r.t. order $\alpha$ (left) and the optimum DP budget w.r.t. sampling probability $q$ (right) of a rPDP-FL algorithm with parameters $T$=20, $\tau$=5, $\lambda$=0.5, $\sigma$=1.0, and $\delta$=1e-3.}
\Description{The DP budget curves w.r.t. order $\alpha$ (left) and the optimum DP budget w.r.t. sampling probability $q$ (right) of a rPDP-FL algorithm with parameters $T$=20, $\tau$=5, $\sigma$=1.0, and $\delta$=1e-3.}
\label{fig:inclusion_prob_est}
\end{figure}

\subsection{Complete Algorithm of rPDP-FL}\label{subsec:compute_rpdpfl}
The SCF strategy has been further integrated into Algorithm \ref{alg:pers_dp_fedavg_pesudo}, with some tweaks in the initialization and sampling procedures. A comprehensive outline of rPDP-FL is provided in Algorithm \ref{alg:pers_dp_fedavg}.

\textbf{The per-record sampling probability initialization.} At the onset of the learning process, the server will compute the sampling probability estimator $F$ and distribute it to all clients. Note that this computation does not rely on personal data from sensitive records stored locally, so there is no risk of compromising the privacy of these records.
On the client side, the per-record sampling probabilities will be calculated by directly plugging in their privacy budgets $\varepsilon_{i,j}$ into the received sampling probability estimator. All clients then employ non-uniform Poisson sampling to randomly select a subset of records based on these probabilities, and apply the DP-SGD algorithm \cite{abadi2016deep} for local model updates.

\textbf{The per-record privacy budget accountant.} 
Another important task for completing rPDP-FL is to keep track of the usage of the privacy budget for each of the records in the course of training. Once the privacy budget runs out, individuals can opt out of the remaining training. In our work, we introduce a monitoring module, called the \textit{budget accountant}, which is in charge of privacy budget accounting: 
(1) \textit{Pre-check} at the beginning of the communication round if an individual has sufficient privacy budget to participate in the current round;
(2) \textit{Compute and update} the accumulated privacy cost of an individual after the current communication round is over. 

\subsection{Discussions} \label{subsec:discussion}

\begin{algorithm}[t] \small
\caption{Record-level Personalized Differentially Private Federated Learning (rPDP-FL, complete version)}
\label{alg:pers_dp_fedavg}
\DontPrintSemicolon
\SetKwInOut{Input}{input}
\SetKwInOut{Parameter}{parameter}
\SetKwInOut{Output}{output}
\Input{$M$ clients with their local datasets $D_{i\in [M]}$ and pre-specified privacy budgets $\{\varepsilon_{i,j}\}_{i\in |D_i|,j\in [M]}$; the total communication rounds $T$, the client-level sampling probability $\lambda$. Parameters shared by all clients:
the local training steps $\tau$; the learning rate $\eta$; the gradient clipping bound $L$, the noise multiplier $\sigma$ and the target DP parameter $\delta$.}
\tcp{Initialization}
$\mathcal{C}\leftarrow$ (all participating clients with size $M$) \;
\color{blue}\tcp{Pre-computation}\color{black}
$F \leftarrow$ (the sampling probability estimator obtained through Alg. \ref{alg:simu_curvefit_strategy}) \;
\ForEach(\emph{in parallel}){$C_i \in \mathcal{C}$ }{
    $\{q_{i,j} = F(\varepsilon_{i,j})\}_{j\in[|D_i|]}\leftarrow$ (the per-record sampling probabilities) \;
}
$\x^0\leftarrow$ (Initialize randomly) \;
\For{$t\in[T]$ }{
    \color{blue}\tcp{Client-level Poisson sampling with the uniform sampling probability $\lambda$}\color{black}
    $\check{\mathcal{C}}^t\leftarrow$ (a random subset drawn from $M$ clients) \;
    \ForEach(\emph{in parallel}){$C_i \in \check{\mathcal{C}}^t$ }{
        $\x_i^{t,0} = \x^t$\;
        
        \For{$r\in[\tau]$ }{
            \color{blue}\tcp{Record-level Poisson sampling  with the derived sampling probability $\{q_{i,j}\}_{j\in [|D_i|]}$}\color{black}
            $S^r\leftarrow$ (a random mini-batch drawn from $D_i$) \;
            \ForEach{microbatch $\xi \in S^r$}{
                $\bar{\g}^r_\xi\leftarrow \grad\ell (\x_i^{t,r}; \xi) \cdot \min(1,\frac{L}{\|\grad\ell (\x_i^{t,r}; \xi)\|_2})$ \;
            }
            $\tilde{\g}^r\leftarrow \frac{\eta}{|S^r|}\left(\sum\bar{\g}^r_\xi+\mathcal{N}(0,\sigma^2  L^2)\right)$ \\
            $\x_i^{t,r+1}\leftarrow\x_i^{t,r}-\eta\tilde{\g}^r$ \;
        }
        $\Delta\x_i = \x_i^{t,\tau} - \x_i^{t,0}$ \;
    }   
    $\x^{t+1}\leftarrow$ (taking the average of all the local updates)\;
}
\end{algorithm}

\noindent\textbf{Generalization of the SCF strategy.} 
A distinct characteristic of the SCF strategy lies in its independence from the inherent complexities of specific processes. While this paper primarily focuses on FL applications, the SCF approach serves as a versatile plug-in module applicable to a broad spectrum of tasks that incorporate data sampling and Gaussian mechanism, including private statistical analysis and other SGD-based optimization tasks.

\noindent\textbf{Sampling/Noise trade-offs.} In highly privacy-sensitive scenarios, the majority of individuals (e.g., patients) prefer stronger privacy protections, and thus their sensitive personal records are less likely to be included in analysis. We argue that this issue is not about our methodology itself, but an inevitable consequence of the personal privacy decision. 
One remedy is to adjust the parameters of the Gaussian mechanism, for example, setting a larger noise multiplier $\sigma$ such that a higher level of Gaussian noise is used in the computation, resulting in universally increased sampling probabilities for everyone. In Figure \ref{fig:inclusion_estimator_sigma}, we illustrate the relationship between optimum DP budgets and their corresponding $q$ (represented as a series of dots) across various $\sigma$ (represented by different colors). Essentially, there is a trade-off between sampling probability and perturbation noise. However, this approach should be treated with extreme caution, as an improper $\sigma$ could lead to a significant degradation in model performance.

\begin{figure}\centering
\includegraphics[width=0.45\linewidth]{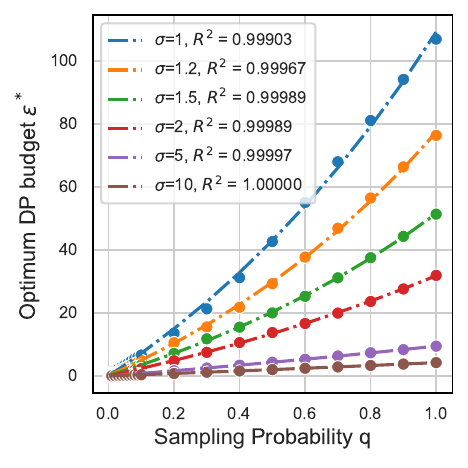}
\caption{The optimum DP budget w.r.t. sampling probability $q$ and noise multiplier $\sigma$ of a rPDP-FL algorithm with parameters $T$=20, $\tau$=5, and $\delta$=1e-5.}
\Description{The optimum DP budget w.r.t. sampling probability $q$ and noise multiplier $\sigma$ of a rPDP-FL algorithm with parameters $T$=20, $\tau$=5, and $\delta$=1e-5.}
\label{fig:inclusion_estimator_sigma}
\end{figure}

\noindent\textbf{The utility risks of the potentially suboptimal $F$.} As a result of the intrinsic traits of the numerical simulation method, the obtained sampling probability estimator $F$ could be suboptimal and lead to utility risks for certain records, e.g., slightly unused budget (due to low sampling probability) or early stop (due to high sampling probability). However, both cases have statistically insignificant effects since the best-fit curve achieves more than 99.9\% $R^2$ value as shown in Figure \ref{fig:inclusion_estimator_sigma}.

\section{Experimental Evaluation} \label{sec:evaluation}

\begin{table*}[t]\scriptsize\centering 
\caption{Overview of the datasets and baseline models used in our experiments.}
\label{tab:datasets}
\begin{tabular}{c|c|c|c|c|c|c|c|c|c|c}
\midrule\midrule
Dataset & Non-/IID & \makecell{\# clients} & \# features & \makecell{\# labels} & \makecell{\# examples\\(per client)} & Train/Test split & Model & \makecell{\# params\\(trainable)} & \makecell{Training steps\\(local $\tau$ / global $T$)} & \makecell{Client-level \\samp. prob. $\lambda$} \\
\hline
\makecell{Heart-Disease \cite{terrail2022flamby}} & 
Non-IID & 4 & \makecell{13} & 2 & 303 / 261 / 46 / 130 & 66\% / 34\% & \makecell{Logistic Regression\\(training from scratch)} & \makecell{20} & \makecell{10 / 15} & 1.0 \\
\hline
\makecell{MNIST \cite{lecun1998gradient}} &
Non-/IID & 10 & 28$\times$28$\times$1  & 10 & $\approx$ 6,000 & 66\% / 34\% 
& \makecell{Two-Layer CNN\\(training from scratch)} & \makecell{26,010} & \makecell{50 / 15} & 0.5 \\
\hline
\makecell{CIFAR10 \cite{krizhevsky2009learning}} &
IID & 10 & 32$\times$32$\times$3 & 10 & 5,000 & 66\% / 34\% & \makecell{ResNet-18 \cite{he2016deep}\\(training from scratch)} & \makecell{11,181,642} & \makecell{50 / 30} & 0.5 \\
\hline
\makecell{SNLI \cite{bowman-etal-2015-large}} &
IID & 10 & \makecell{Premise-\\hypothesis\\pairs} & 3 & 54,936 & 95\% / 5\% &  \makecell{Pretrained BERT \cite{kenton2019bert}\\(fine-tuning)} & \makecell{7,680,771} & \makecell{50 / 15} & 0.5 \\
\midrule\midrule
\end{tabular}
\end{table*}

In this section, we conduct a thorough empirical analysis to evaluate the performance of both the SCF strategy and the rPDP-FL algorithm. 
In Subsection \ref{subsec:existing_vs_ours}, we focus on evaluating the effectiveness of SCF by comparing it with the other two existing strategies employed for achieving PDP in centralized ML settings. Given the absence of alternative implementations achieving record-level PDP in the context of FL, we assess the utility improvement of rPDP-FL by contrasting it with two conventional methods that do not incorporate personalized privacy preservation in Subsection \ref{subsec:uniform_vs_ours}. The source code, data, and other artifacts have been made available\footnote{https://github.com/Emory-AIMS/rPDP-FL.git}.
%\footnote{https://github.com/JunxuLiu/rPDP-FL.git}. 

\noindent
\textbf{Privacy Preference Distributions.} We simulate different scenarios where users have diverse privacy preferences for their data. 
\begin{itemize}[leftmargin=*]
\item \textbf{ThreeLevels}: each record has the option to select a preferred privacy budget from three distinct choices (categories) $\{\varepsilon_1,\varepsilon_2,\varepsilon_3\}$ where $\varepsilon_1  < \varepsilon_2 < \varepsilon_3$, each denoting strong, moderate, and weak privacy protection, respectively. This simulates practical PDP implementation scenarios where the users can choose from a few predefined privacy categories. 
\item \textbf{BoundedPareto}: each record has an arbitrary $\varepsilon_{i,j} \in [0.1, 10]$ that approximately follows a Pareto distribution. This simulates the scenario where a majority of individuals lean towards stringent privacy safeguards, while a smaller subset opts for less restrictive protections in exchange for improved services or other incentives.
\item \textbf{BoundedMixGauss}: each record has an arbitrary $\varepsilon_{i,j} \in [0.1, 10]$ that approximately follows a mixture of three Gaussian distributions  with means $\{\varepsilon_1,\varepsilon_2,\varepsilon_3\}$ where $\varepsilon_1  < \varepsilon_2 < \varepsilon_3$.  This simulates the scenarios where the privacy choices are multi-modal as in many other complex social systems \cite{WuPNAS18803, scalas2005intertrade,scalas2004anomalous}. 
\end{itemize}

\noindent
\textbf{Datasets and Models.} We consider four classification tasks with the consistent objective of training a global model privately on the federated \textit{Heart Disease} \cite{terrail2022flamby}, \textit{MNIST} \cite{lecun1998gradient}, \textit{CIFAR10} \cite{krizhevsky2009learning}, and \textit{SNLI} \cite{bowman-etal-2015-large}, separately. Note that \textit{Heart-Disease} is a real healthcare dataset comprising records from 920 patients across four hospitals in Cleveland, Hungary, Switzerland, and Long Beach V. 
On the other hand, \textit{MNIST} and \textit{CIFAR10} are two commonly used benchmarks for image classification tasks, while the \textit{SNLI} dataset is a benchmark for natural language inference (NLP) tasks. In these cases, we apply the IID and non-IID partitioning strategies introduced in \cite{mcmahan2017communication} to split total training examples into $M=10$ subsets. For a more comprehensive overview of the datasets, along with details on the corresponding baseline models, please refer to Table \ref{tab:datasets}. 

\noindent\textbf{Implementations.} Our implementation utilizes the Opacus library. All experiments are conducted on a machine with one NVIDIA A40 GPU running on Ubuntu with 256 GB memory. Given that the model training is a randomized process, we repeat all the experiments five times and report the mean test accuracy across all clients.

\subsection{Comparison of SCF with Existing Strategies}\label{subsec:existing_vs_ours}

In this section, we show the effectiveness and efficiency of our SCF strategy in terms of model utility and computational cost by comparing it with the following representative approaches:
\begin{itemize}[leftmargin=*]
\item \textbf{Filter}: also known as R\'enyi privacy filter (Algorithm 3 in \cite{feldman2021individual}), is an individual privacy accounting method that monitors the accumulation of \textit{squared gradient norms} $B_{norm}$ for each record during the training process. The record will be filtered out if this accumulation exceeds a pre-specified threshold.
\item \textbf{BinarySearch}: also known as Individual DP-SGD (IDP-SGD) with the Sample mechanism (Algorithm 2 in \cite{boenisch2023have}), is a \textit{binary search}-based approach aiming for finding the optimal sampling probability within the range of $[0,1]$ for a target privacy budget.
\end{itemize}

Given that both approaches were initially tailored for \textit{centralized} ML scenarios, our experiments adhere to this context to maintain fairness in comparisons. Specifically, we implement a variant of the DP-SGD \cite{abadi2016deep} algorithm, incorporating refinements in the predetermination of record-level sampling probabilities through the SCF strategy. Note that here the RDP budget curve (line 4 in Alg. \ref{alg:simu_curvefit_strategy}) should be calculated based on Lemma \ref{lemma:rdp_amplification} instead of Theorem \ref{theorem:indiv_priv_analysis}.

\subsubsection{Comparison of SCF with Filter}\label{subsubsec:scf_vs_filter} While both \textit{SCF} and \textit{Filter} share a common objective of achieving personalized privacy protection, they significantly differ in the definition of ``budgets'': \textit{Filter} considers a budget for the accumulative squared gradient norms for the records in the training process, while we focus on the DP privacy budget $\varepsilon$ for the accumulative privacy cost. Achieving a smooth transition between the two quantities is challenging\footnote{We highlight that the final privacy budget $\varepsilon$ depends on the sampling probability $q$, noise multiplier $\sigma$, and the number of training steps $\tau$. However, in \textit{Filter}, different records could have varying values of $\tau$ even though they share the same ``squared norm budget'' $B_{norm}$. As a result, adjusting the hyperparameters of $B_{norm}$ to align with a specific $\varepsilon$ for each record, and vice versa, becomes challenging, particularly in general personalized privacy scenarios like BoundedMixGauss or BoundedPareto.}.
Therefore, our evaluation only focuses on a simplified personalized privacy scenario \textit{ThreeLevels}, where the percentage of records with $\varepsilon_1$, $\varepsilon_2$, and $\varepsilon_3$ is 70\%, 20\%, and 10\%, respectively. We incorporate \textit{Filter} into the private gradient descent algorithm\footnote{It computes noisy gradients using the entire training dataset in each iteration, which is different from DP-SGD.} as implemented in \cite{feldman2021individual}. Then we assess the model utility and computational efficiency of both approaches on the pooled Heart-Disease and MNIST datasets. 

\begin{table}[tbp]\scriptsize\centering 
\caption{Comparison of SCF with Filter in centralized ML.}
\label{tab:filter_vs_scf}
\begin{tabular}{c|c|c|c|c|c}
\midrule\midrule
\makecell{Priv. Pref. Dist.} & Dataset & Method & Batch Size & Test Acc & \makecell{Runtime\\(s/step)} \\
\hline
\multirow{4}{*}{\makecell{ThreeLevels:\\$\varepsilon_1\approx$ 2.0 (70\%)\\$\varepsilon_2\approx$ 4.7 (20\%)\\$\varepsilon_3\approx$ 11.8 (10\%)}} & 
\multirow{2}{*}{\makecell{Heart-Disease\\(pooled)}} & 
Filter & 486 & 0.7717 & 0.0128 \\
\cline{3-6}
& & SCF & 32 (Expected) & \textbf{0.8189} & \textbf{0.0050} \\
\cline{2-6}
& 
\multirow{2}{*}{\makecell{MNIST\\(pooled)}} & Filter & 60000 & 0.8411 & 7.1016 \\
\cline{3-6}
& & SCF & 64 (Expected) & \textbf{0.9477} & \textbf{2.1824} \\
\midrule\midrule
\end{tabular}
\end{table}

\begin{figure}[tbp]\centering
\includegraphics[width=\linewidth]{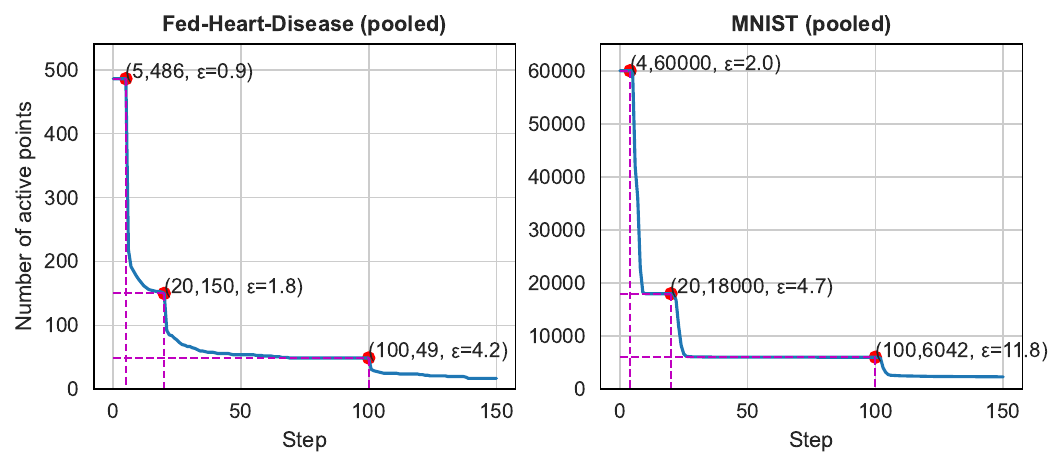}
\caption{Number of active records during one run of private GD with Filter evaluated on the pooled Fed-Heart-Disease dataset (left) and MNIST (right). The red points indicate the specific steps that records with different privacy budgets start to get filtered.} 
\Description{Number of active records during one run of private GD with Filter evaluated on the pooled Fed-Heart-Disease dataset (left) and MNIST (right). The red points indicate the specific steps that records with different privacy budgets start to get filtered.} 
\label{fig:filter_active_points}
\end{figure}

\begin{table}[tbp]\scriptsize\centering 
\caption{Comparison of SCF with BinarySearch in centralized ML.}
\label{tab:binary_vs_scf}
\begin{tabular}{c|c|c|c}
\midrule\midrule
\makecell{Priv. Pref. Dist.} & Method & Test Acc & \makecell{Runtime (s)} \\
\hline
\multirow{2}{*}{\makecell{Group-3: 3 unevenly sized privacy groups\\(54\%-37\%-9\%) with privacy budgets [1., 2., 3.]}} &  
BinarySearch & 0.7274 & \textbf{1.69} \\
\cline{2-4}
& SCF & 0.7240 & 13.11 \\
\hline
\multirow{2}{*}{\makecell{Group-100: 100 evenly sized privacy groups\\with privacy budgets [1., 1.05, ..., 5.95]}} & BinarySearch & 0.8135 & 52.50 \\
\cline{2-4} 
& SCF & 0.8134 & \textbf{13.32} \\
\hline
\multirow{2}{*}{\makecell{Individual-1000: per-record privacy budgets\\drawn from BoundedMixGauss}} & BinarySearch & 0.6858  & 597.82 \\
\cline{2-4} 
& SCF & 0.6861 & \textbf{14.10} \\
\midrule\midrule
\end{tabular}
\end{table}

\textbf{Number of active records.} In Figure \ref{fig:filter_active_points}, we plot the number of active records during one run of private GD with \textit{Filter}, i.e. those that have not yet exhausted their privacy budgets. The results reveal that in scenarios where personalized privacy protection is implemented, particularly when a significant portion of the records adheres to conservative privacy preferences, a considerable number of records undergo early filtration in the training process, which restricts the model's capacity to learn from the dataset effectively. As shown in Table \ref{tab:filter_vs_scf}, our experimental results demonstrate that our \textit{SCF} outperforms \textit{Filter} in terms of test accuracy.

\textbf{Computational efficiency and privacy amplification.}
Due to the requisite computation of per-example gradient norms at each iteration, \textit{Filter} no longer retains the benefits of SGD for improving speed and memory efficiency \cite{feldman2021individual}. 
In contrast, the Poisson sampling procedure involved in our approach could yield mini-batches of size $\sum_{j\in |D|}{q_j} < |D|$ in expectation. We report the average time cost for each training step in Table \ref{tab:filter_vs_scf}, revealing SCF achieves a 2x or 3x speedup compared to Filter in our assessments.
Moreover, the absence of the random data subsampling procedure in private GD with Filter also leads to the loss of privacy amplification effect for individual records. 

\subsubsection{Comparison of SCF with BinarySearch}
As discussed in Section \ref{sec:related_work}, BinarySearch faces a significant limitation in terms of efficiency when dealing with records' privacy budgets that cover values within a continuous range, rather than a discrete subset. To illustrate this, we carry out a small-scale experiment where a private model is trained on a subset of 1,000 examples randomly selected from the MNIST dataset using DP-SGD. Three different personalized privacy scenarios are being investigated to evaluate the efficiency and utility of BinarySearch and SCF, as outlined in Table \ref{tab:binary_vs_scf}. Note that both Group-3 and Group-100 are privacy setups concerned in \cite{boenisch2023have}, where all records are split into limited privacy groups and those within one privacy group share the same privacy budget. The Individual-1000 represents the privacy scenario considered in this paper, where the privacy budgets are assigned on an individual basis. 
We report the test accuracies and the runtimes to obtain the per-record sampling probabilities (averaged over 5 trials). Our experiment results demonstrate that SCF significantly enhances efficiency compared to BinarySeach in the more general scenarios while achieving a comparable model performance. 
Here we did not evaluate BinarySearch in a parallel manner as this implementation was not discussed in \cite{boenisch2023have} and is not the focus of this paper.

\begin{remark}
Boenisch et al. \cite{boenisch2023have} also proposed a Scale mechanism that aims at scaling the noise added to each gradient by setting individualized clipping bounds. The optimal clipping bound is still found in a binary-search style. According to Table 16-17 in \cite{boenisch2023have}, Scale leads to a comparable overall performance with Sample (regarding runtime and test accuracy). Thus, both Sample and Scale share similar limitations against our method.
\end{remark}

\begin{figure*}[ht]\centering
\includegraphics[width=\linewidth]{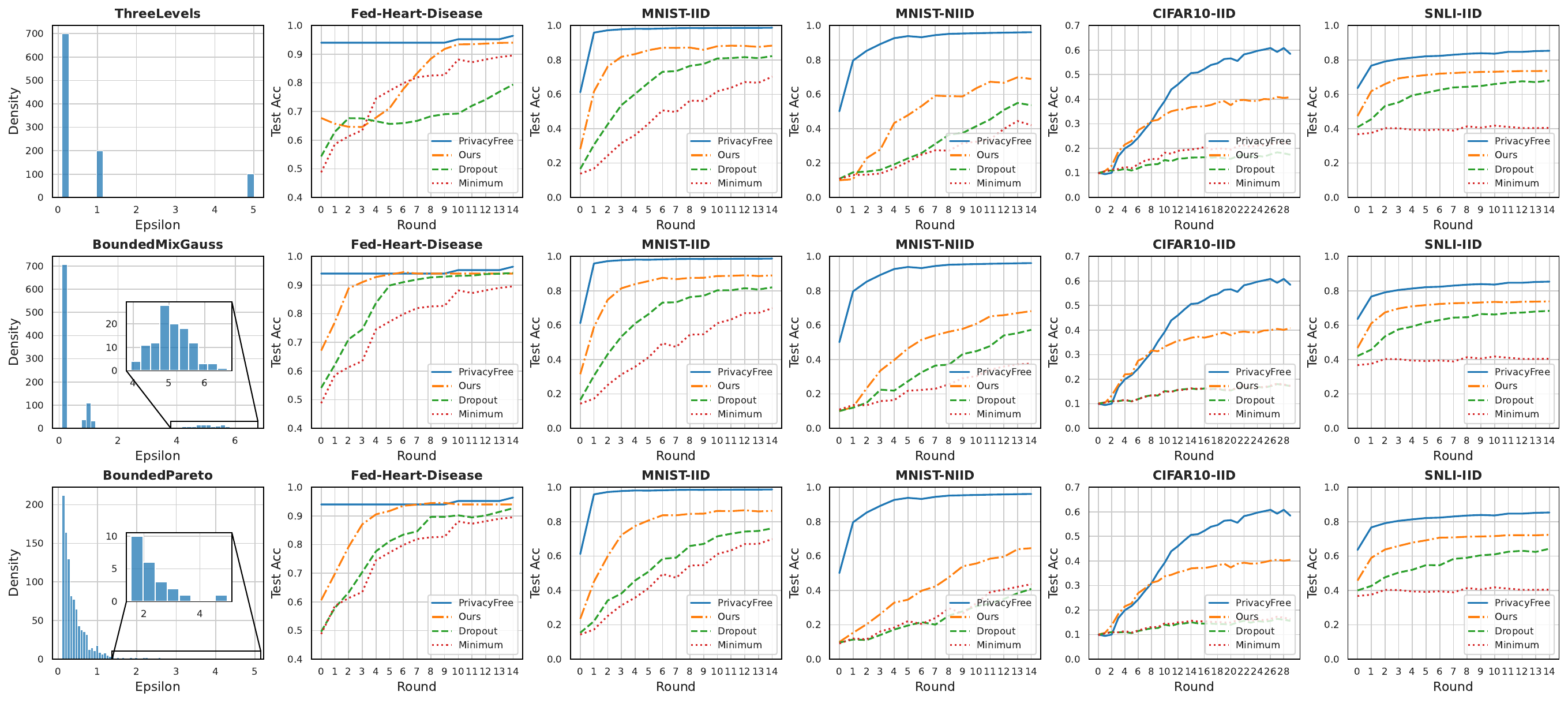}
\caption{Evaluation of rPDP-FL (labeled as \textit{Ours}) across diverse privacy preference distributions and datasets in federated learning.} 
\Description{ Evaluation of rPDP-FL (labeled as \textit{Ours}) across diverse privacy preference distributions and datasets in federated learning.} 
\label{fig:all_plots}
\end{figure*}

\subsection{Utility Improvement through Privacy Personalization}\label{subsec:uniform_vs_ours}
In this section, we evaluate the effectiveness of our rPDP-FL algorithm in providing record-level personalized DP while optimizing for model accuracy. Given the absence of alternative implementations achieving record-level PDP in the FL setting, we  established three DP-FedAvg-based baseline methods as follows:
\begin{itemize}[leftmargin=*]
\item \textbf{Minimum}: DP-FedAvg ensuring the most stringent \textit{uniform} ($\varepsilon_{min}, \delta$)-DP guarantees, where $\varepsilon_{min} = \min_{i\in [M], j\in [|D_i|]} \varepsilon_{i,j}$. 
\item \textbf{Dropout}: DP-FedAvg providing moderate \textit{uniform} ($\varepsilon_{mod},\delta$)-DP guarantees for a subset of individuals whose privacy budgets $\varepsilon_{i,j}$ are larger than a predefined threshold $\varepsilon_{mod}$ and dropping out those with privacy budgets below the threshold. Here, $\varepsilon_{mod}$ is determined as the empirical sample mean of $\frac{1}{|\mathcal{D}|} \sum_{i,j} \varepsilon_{i,j}$.
\item \textbf{PrivacyFree}: The vanilla FedAvg \cite{mcmahan2017communication} without DP guarantees, which serves as a benchmark for assessing the reduction in global model utility.
\end{itemize}
Note that both \textit{Minimum} and \textit{Dropout} ensure the accumulative privacy costs of all records remain at/below their specified privacy budgets throughout the training process but lead to a significant waste for records with large privacy budgets. All three types of privacy preference distributions are considered for a comprehensive evaluation of all methods' performance.

Unless otherwise specified, the default setup of \textit{ThreeLevels} comprises $\varepsilon_1=0.1$, $\varepsilon_2=1.0$ and $\varepsilon_3=5.0$ with corresponding proportion of 70\%, 20\%, and 10\%. Similarly, the \textit{BoundedMixGauss} distribution is defined as a mixture of $\mathcal{N}_1(0.1, 0.01)$, $\mathcal{N}_2(1.0, 0.05)$ and $\mathcal{N}_3(5.0, 0.5)$ with weights of 0.7, 0.2 and 0.1. 
For the \textit{BoundedPareto} case, we consider a specific Pareto distribution with a shape value of 1.0 and a lower bound of 0.1. These parameters are chosen to simulate realistic scenarios where the majority of users have strict privacy requirements. 
Examples of 1,000 records’ privacy preferences are illustrated in the first column of Figure \ref{fig:all_plots}.
 In the context of training a larger ResNet-18 model from scratch on the CIFAR-10 dataset, where the total number of trainable parameters exceeds 11M, we opt for more liberal privacy settings to maintain acceptable model utility. Specifically, we set \(\varepsilon_1=1.0\), \(\varepsilon_2=3.0\), and \(\varepsilon_3=10.0\) for both the \textit{ThreeLevels} and \textit{BoundedMixGauss} setups while keeping the ratios the same. For the \textit{BoundedPareto} distribution, the lower bound is adjusted to 1.0.

\textbf{Hyperparameters.} For the Fed-Heart-Disease experiments, we fix the per-client sampling probability $\lambda$=1.0 and $\delta$=1e-3; for the other experiments, we fix $\lambda$=0.5 and $\delta$=1e-4. The total communication rounds $T$ and the number of local training steps per round $\tau$ for different tasks are detailed in Table \ref{tab:datasets}. %(e.g., the level of additive Gaussian noise, the norm clipping bound of gradients $L$, the value of $\delta$, etc)
We explore a variety of constant learning rates (e.g., [0.1, 0.05, 0.01, 0.005, 0.001]) and clipping thresholds (e.g., [0.5, 1.0, 3.0, 5.0]) for the best results. 

\textbf{Model Utility.} 
The evaluation results under various privacy preference distributions, as shown in Figure \ref{fig:all_plots}, consistently demonstrate the superiority of our \textit{rPDP-FL} method over both \textit{Minimum} and \textit{Dropout} across diverse FL datasets. Specifically, the advantages over the \textit{Minimum} method suggest that the model benefits from records with larger privacy budgets, while its advantages over the \textit{Dropout} method imply that even records with conservative privacy budgets contribute positively to the training process. These findings underscore the crucial role of personalized privacy integration in achieving a more favorable trade-off between privacy and utility, particularly in scenarios where most individuals exhibit a strong emphasis on privacy concerns.

Upon comparing the results between MNIST-IID and MNIST-NIID, we observed a consistent reduction in test accuracy in the non-IID case across \textit{Minimum}, \textit{Dropout}, and our proposed method. This decline arises from the inherent limitations of the \textit{DP-FedAvg} algorithm under a heterogeneous environment \cite{noble2022differentially}. Despite this deterioration, it's noteworthy that the utility advantages stemming from incorporating privacy personalization remain. 

\section{Conclusion and Future Work}\label{sec:conclusion}

In this paper, we studied an unexplored real-world challenge to enable record-level personalized differential privacy in federated learning. Our proposed solution is a novel framework called \textit{rPDP-FL}, which employs a \textit{two-stage hybrid sampling} scheme with non-uniform record-level sampling. We devise an efficient strategy, referred to as \textit{Simulation-CurveFitting} (SCF), to estimate the individual sampling probabilities for all records associated with varying privacy budgets. Our investigation uncovers a valuable insight regarding rPDP-FL, i.e., a one-to-one correspondence between the sampling probabilities of records and their accumulative privacy costs which can be mathematically expressed through a simple exponential function. Empirical demonstrations show that our solutions yield significant performance enhancements compared to baselines that overlook personalized privacy preservation. 

As an early exploration into FL with personalized privacy protection, our work lays a foundation for future in-depth investigations and highlights several promising directions, such as: 
\begin{itemize}[leftmargin=*]
\item \textbf{User-level (device-level) PDP in cross-silo (cross-device) FL}. 
In FL scenarios, individual users (or devices) might possess multiple records or contribute data to various clients simultaneously \cite{kato2023uldp}. Intuitively, a single user may have distinct privacy preferences for their records, which poses significant challenges in establishing user-level personalized privacy protection.
\item \textbf{Effective learning on non-IID data with data-dependent privacy budgets.} 
In Appendix \ref{appendix}, we offer a brief discussion on a more complex scenario where the privacy budgets of users are directly linked to their raw data (e.g., dependent on their labels). Our findings point out the deficiencies in the current privacy personalization methods to yield substantial utility gain for groups that have significantly smaller privacy budgets and are a minority in the population. We highlight the need to address this challenge effectively in future research endeavors.
\end{itemize}

\begin{acks}
We would like to thank the reviewers for their thoughtful comments and efforts toward improving our paper and artifacts. 
This work was supported in part by the National Natural Science Foundation of China grants (62172423, 62206207, 62102352, U23A20306), National Institutes of Health grants (R01ES033241, R01LM013712), National Science Foundation grants (CNS-2124104, CNS-2125530, IIS-2302968), and National Key Research and Development Program of China grant (2021YFB3101100).
\end{acks}

\bibliographystyle{ACM-Reference-Format}%\footnotesize
\bibliography{main_references}

\appendix
\addcontentsline{toc}{section}{Appendices}
\balance
\begin{figure}[htbp]\centering
\includegraphics[width=0.5\linewidth]{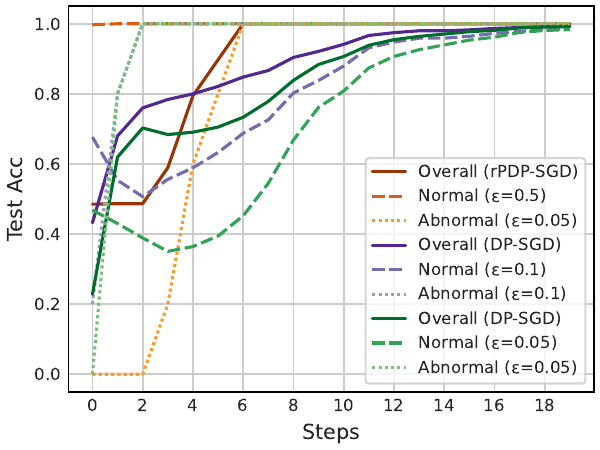}
\caption{Evaluation of the per-class test accuracy of a logistic regression model trained on the Heart-Disease dataset with data-dependent privacy budgets ($\varepsilon$=0.5 for normal patients and $\varepsilon$=0.05 for abnormal patients) for 20 iterations.}
\Description{Evaluation of the per-class test accuracy of a logistic regression model trained on the Heart-Disease dataset with data-dependent privacy budgets ($\varepsilon$=0.5 for normal patients and $\varepsilon$=0.05 for abnormal patients) for 20 iterations.}
\label{fig:heart_per_class_budget_assignment}
\end{figure}

\begin{table*}[t]\small\centering
\caption{Evaluation of the per-class test accuracy of a CNN model trained on MNIST  with data-dependent privacy budgets for 100 iterations.}
\label{tab:mnsit_per_class_budget_assignment}
\begin{tabular}{c|c|c|c|c|c|c|c|c|c|c|c}
\midrule\midrule
MNIST & 0 & 1 & 2 & 3 & 4 & 5 & 6 & 7 & 8 & 9 & Overall \\
\hline
Budget ($\varepsilon$) & 0.5 & 0.75 & 2.0 & 2.6 & 4.1 & 2.1 & 2.05 & 3.0 & 3.1 & 6.1 & - \\
\hline
rPDP-SGD & 93.49 & 95.9 & 87.45 & \textbf{90.83} & \textbf{95.1} & \textbf{90.37} & \textbf{94.74} & \textbf{92.13} & \textbf{88.62} & \textbf{92.64} & \textbf{92.55} \\
\midrule[1pt]
DP-SGD ($\varepsilon$=0.5) & 93.63 & 95.81 & 87.26 & 87.18 & 89.67 & 85.89 & 93.73 & 88.71 & 85.74 & 82.19 & 89.08\\
\hline
DP-SGD ($\varepsilon$=3.0) & \textbf{94.24} & \textbf{96.09} & \textbf{89.75} & 89.71 & 92.01 & 89.37 & 94.54 & 90.5 & 84.78 & 85.2 & 90.69 \\
\midrule[1pt]
Vanilla SGD & 99.1 & 99.6 & 97.11 & 98.33 & 99.1 & 98.65 & 98.27 & 99.07 & 95.87 & 94.1 & 98.06\\
\midrule\midrule
\end{tabular}
\end{table*}

\section{Additional Evaluation Results}
\label{appendix}

\noindent\textbf{Learning with data-dependent privacy bdugets.}
In this study, we focus on a general personalized privacy scenario in which individual privacy budgets follow a random distribution, independent of other factors such as the raw data. Nonetheless, there are scenarios where individuals with particular attributes (or labels) might exhibit different privacy concerns. For example, heart disease patients might demand more stringent privacy safeguards when their health records are utilized in training ML models. 

To examine the utility improvement through privacy personalization on different groups, we carry out preliminary experiments by training centralized ML models on the Heart-Disease and MNIST datasets. In each experiment, we apply DP-SGD and a variant of DP-SGD that incorporates the SCF strategy for achieving rPDP, which we denoted as \textit{rPDP-SGD}. We allocate distinct privacy budgets to each class. For the Heart-Disease dataset, ``normal'' patients are assigned a privacy budget of $\varepsilon$=0.5, whereas ``abnormal'' records are given a more conservative budget of $\varepsilon$=0.05. For the MNIST dataset, we adhere to the setup described in \cite{boenisch2023have}. 

In Figure \ref{fig:heart_per_class_budget_assignment}, we visualize per-class and overall test accuracy (averaged over 10 trials) for the logistic regression model trained on the Heart-Disease dataset. Due to the inherent simplicity of the dataset, both \textit{rPDP-SGD} and baselines ($\varepsilon$=0.05 and 0.1 for all classes) achieve convergence to perfect accuracy (1.0) within 20 iterations. In this experiment, we do not observe a discernible utility gain of rPDP-SGD compared to the baselines for the ``abnormal'' records (as indicated by the three dotted lines in the figure). This lack of utility improvement could stem from the unbalancedness of data distribution, together with the fact that ``abnormal'' records are sampled less frequently than ``normal'' records, causing the model to primarily learn from ``normal'' records during the initial phases of training.

Table \ref{tab:mnsit_per_class_budget_assignment} displays the final test accuracy (averaged over 5 trials) for 10 evenly sized classes of the MNIST dataset. It can be observed that for classes with privacy budgets below 3.0 (Classes 3, 5, 6, and 7), their test accuracy of rPDP-SGD significantly surpass those of the other two baselines. For classes with much lower privacy budgets (Classes 0, 1, and 2), rPDP-SGD demonstrates performance on par with that of DP-SGD ($\varepsilon$=3.0). 
Our findings indicate that the sampling-based method does not yield substantial utility improvements for groups that have significantly smaller privacy budgets and are a minority in the population. As discussed in Section \ref{subsec:existing_vs_ours}, current methods such as Filter and BinarySearch also fail to adequately address this issue. This suggests that an efficient and effective solution for this challenge has yet to be developed, leaving it as an open question for further investigation.

\end{document}